%% file: main.tex
\documentclass{article}
\input{head}

\title{Fingerprint Filters Are Optimal}
\author{
  William Kuszmaul\thanks{Partially supported by NSF grant CNS2504471. \texttt{kuszmaul@cmu.edu}.} \\
  CMU
  \and
  Jingxun Liang\thanks{\texttt{jingxunl@andrew.cmu.edu}.} \\
  CMU
  \and
  Renfei Zhou\thanks{Partially supported by the Jane Street Graduate Research Fellowship and the MongoDB PhD Fellowship. \texttt{renfeiz@andrew.cmu.edu}.} \\
  CMU
}
\date{}

\begin{document}
\date{}
\maketitle

\begin{abstract}
    Dynamic filters are data structures supporting approximate membership queries to a dynamic set $S$ of $n$ keys, allowing a small false-positive error rate $\epsilon$, under insertions and deletions to the set $S$. Essentially all known constructions for dynamic filters use a technique known as \emph{fingerprinting}. This technique, which was first introduced by Carter et al.~in 1978, inherently requires $$\log \binom{n \epsilon^{-1}}{n} = n \log \epsilon^{-1} + n \log e - o(n)$$ bits of space when $\epsilon = o(1)$. Whether or not this bound is optimal for \emph{all} dynamic filters (rather than just for fingerprint filters) has remained for decades as one of the central open questions in the area. We resolve this question by proving a sharp lower bound of $n \log \epsilon^{-1} + n \log e - o(n)$ bits for $\epsilon = o(1)$, regardless of operation time.
    
\end{abstract}

\section{Introduction}
\label{sec:introduction}
\input{introduction}

\section{Preliminaries}
\label{sec:preliminaries}
\input{preliminaries}

\section{Warmup: Lower Bound Under Two Assumptions}
\label{sec:warmup}
\input{warmup}

\section{Beyond History Independence}
\label{sec:history_dependent}
\input{history_dependent}

\section{Beyond Monotonicity}
\label{sec:non-monotone}
\input{non-monotone}

\bibliographystyle{alpha}
\bibliography{reference.bib}

\end{document}

%% file: head.tex
\usepackage[T1]{fontenc}
\usepackage[margin=1in]{geometry}
\usepackage{amsfonts,amsmath,amsthm,amssymb,mathtools}  %
\usepackage{todonotes}

\usepackage{nicematrix}

\mathtoolsset{centercolon}
\usepackage{xfrac,nicefrac}
\usepackage{mathdots}
\usepackage{mleftright}  %

\usepackage{xspace}
\xspaceaddexceptions{]\}}  %
\usepackage{regexpatch}

\usepackage{bm,bbm,dsfont}  %
\usepackage{caption}
\usepackage[normalem]{ulem}
\usepackage{enumitem}

\usepackage{graphicx}
\usepackage{float}
\usepackage{subcaption}  %
\usepackage{tcolorbox}
\usepackage{tikz}
\usetikzlibrary{decorations.pathreplacing}
\usetikzlibrary{calc}
\usetikzlibrary{positioning}
\usetikzlibrary{arrows.meta}
\usetikzlibrary{math}
\usetikzlibrary{patterns}
\usetikzlibrary{trees}

\usetikzlibrary{trees}

\usepackage[linesnumbered,boxed,ruled,vlined]{algorithm2e}

\SetCommentSty{mycommfont}

\usepackage{thmtools,thm-restate}
\usepackage[colorlinks,citecolor=blue,linkcolor=blue,urlcolor=red]{hyperref}

\theoremstyle{plain}

\newtheorem{theorem}{Theorem}[section]  %
\newtheorem{lemma}[theorem]{Lemma}

\newtheorem{proposition}[theorem]{Proposition}

\newtheorem{claim}[theorem]{Claim}

\theoremstyle{definition}  %
\newtheorem{definition}[theorem]{Definition}

\newenvironment{proofof}[1]{\begin{proof}[Proof of #1]}{\end{proof}}

\usepackage[capitalise]{cleveref}
\crefname{algocf}{Algorithm}{Algorithms}
\Crefname{algocf}{Algorithm}{Algorithms}
\crefname{claim}{Claim}{Claims}
\Crefname{claim}{Claim}{Claims}

\newfloat{Distribution}{htbp}{loa}
\crefname{Distribution}{Distribution}{Distributions}
\Crefname{Distribution}{Distribution}{Distributions}
\newfloat{Protocol}{htbp}{loa}
\crefname{Protocol}{Protocol}{Protocols}
\Crefname{Protocol}{Protocol}{Protocols}

\SetKwProg{Function}{Function}{:}{}
\let\Fn\Function
\SetKwProg{CodeBlock}{}{}{}
\SetKwProg{Repeat}{Repeat}{:}{}

\DeclarePairedDelimiter{\bk}{(}{)}
\DeclarePairedDelimiter{\Bk}{[}{]}
\DeclarePairedDelimiter{\BK}{\{}{\}}

\DeclarePairedDelimiter{\abs}{\lvert}{\rvert}

\DeclarePairedDelimiterX\mysetbase[2]{\lbrace}{\rbrace}{#1\,\delimsize\vert\,#2}
\NewDocumentCommand{\myset}{sO{}m m}{%
  \IfBooleanTF{#1}%
    {\mysetbase*{#3}{#4}}%
    {\mysetbase[#2]{#3}{#4}}%
}

\DeclareMathOperator*{\E}{\mathbb{E}}

\let\Pr\PrAux

\DeclareMathOperator*{\ind}{\mathbbm{1}}

\renewcommand{\tilde}{\widetilde}

\newcommand{\defeq}{\coloneqq}
\newcommand{\eps}{\varepsilon}

\newcommand{\N}{\mathbb{N}}
\newcommand{\R}{\mathbb{R}}

\renewcommand{\l}{\ell}

\renewcommand{\epsilon}{\eps}

\newcommand{\defn}[1]{\emph{\boldmath\textbf{#1}}}

\newcommand{\numberthis}{\addtocounter{equation}{1}\tag{\theequation}}

\usepackage{regexpatch}
\makeatletter
\xpatchcmd\thmt@restatable{%
\csname #2\@xa\endcsname\ifx\@nx#1\@nx\else[{#1}]\fi
}{%
\ifthmt@thisistheone
\csname #2\@xa\endcsname\ifx\@nx#1\@nx\else[{#1}]\fi
\else
\csname #2\@xa\endcsname[{Restated}]
\fi}{}{}
\makeatother

\usepackage{physics}

\newcommand{\accept}[1]{\overline{#1}}
\newcommand{\acceptvar}[1]{\tilde{#1}}
\newcommand{\hfilter}{H_{\textup{filter}}}
\newcommand{\calD}{\mathcal{D}}
\newcommand{\send}[1]{\texttt{Send(}#1\texttt{)}}
\newcommand{\tapeD}{\text{tape}(\calD)}

%% file: introduction.tex
\defn{Filters} are space-efficient data structures designed for \emph{approximate set membership queries}, allowing a small false-positive error rate. Specifically, given a key set $S \subset U$ of at most $n$ keys, a filter with a false-positive error rate $\eps$ supports a \texttt{Query} operation for any $x \in U$ with the following guarantees:  
\begin{itemize}
    \item $\texttt{Query(}x\texttt{)}$ returns \texttt{true} for each $x \in S$.  
    \item $\texttt{Query(}x\texttt{)}$ returns \texttt{false} with probability at least $1 - \eps$ for each $x \notin S$.  
\end{itemize}  
In other words, queries for elements in $S$ are exact (always returning \texttt{true}), while queries for elements outside $S$ are approximate (most likely returning \texttt{false}, but sometimes producing false positives).

While similar to exact membership data structures like hash tables in functionality, filters require significantly less space, making them practical for various real-world applications (see, e.g., the surveys \cite{broder2004network, abdennebi2021bloom, tarkoma2012theory, patgiri2018role, luo2019optimizing, patgiri2018preventing, blustein2002bloom, nayak2019review, singh2020probabilistic}).

In this paper, we focus on \defn{dynamic filters}, which are filters that additionally support insertions and deletions to the key set $S$. In this setting, $n$ serves as the \emph{maximum capacity}, representing an upper bound on the number of elements in $S$ at any given time. 

The construction of dynamic filters in the literature has been dominated by a single paradigm: \defn{fingerprint filters}, introduced by \cite{carter1978exact}. In a fingerprint filter, each key is hashed to a random \emph{fingerprint} in the set $\{1, 2, \ldots, n \epsilon^{-1}\}$.\footnote{Slightly better bounds can be achieved by using a set of size very slightly less than $n \epsilon^{-1}$ for the fingerprints. However, so long as $\epsilon = o(1)$, this distinction only ends up affecting the final space usage by an additive $o(n)$ bits.} The data structure maintains a succinct hash table for all fingerprints and answers queries based on whether the fingerprint of the queried key is present in the hash table. This framework has been widely used in both theoretical dynamic filter works \cite{carter1978exact, pagh2005optimal, arbitman2010backyard, bender2022optimal,bender2018bloom,liu2020succinct,bercea2020dynamic} and in practical implementations such as quotient filters \cite{clerry1984compact, pagh2005optimal, bender2012dont, pandey2021vector, pandey2017generalpurpose, geil2018quotient} and cuckoo filters \cite{pagh2004cuckoo, fan2014cuckoo, chen2017dynamic, luo2019consistent}. The state of the art on the theory side is due to Bender et al. \cite{bender2022optimal}, who construct dynamic filters achieving $O(1)$-time operations using $n \log \epsilon^{-1} + n \log e + o(n)$ bits of space\footnote{The operator $\log$ denotes the binary logarithm in this paper.}, as long as $\log \epsilon^{-1} \in \omega(1) \cap O(\log n / (\log \log \cdots \log n))$ (for any constant number of logarithms). This result comes close to the theoretical limit for fingerprint filters, since fingerprint filters inherently require $\log \binom{n \eps^{-1}}{n} = n\log \eps^{-1} + n \log e - o(n)$ bits of space when $\eps = o(1)$.

What is not known is whether fingerprint filters are \emph{optimal}. Could a dynamic filter---implemented using a different paradigm---hope to break through the $n \log \epsilon^{-1} + n \log e - o(n)$-bit barrier? This simple question has stood as one of the central open questions in the area for nearly five decades \cite{carter1978exact}. 

It has been known since the late 1970s \cite{carter1978exact} that even static filters must use at least $n \log \epsilon^{-1}$ bits of space. The first separation between static and dynamic filters was given by Lovett and Porat \cite{lovett2013space}, who established a lower bound of the form $n \log \eps^{-1} + n f(\eps)$ for some positive function $f$.\footnote{Notably, their lower bound holds even for \emph{incremental filters} (i.e., filters that only support insertions). It was subsequently shown by \cite{kuszmaul2024space} that this type of lower bound is the best one can hope for in the incremental setting, as there exist incremental filters using $n \log \eps^{-1} + o(n)$ bits when $\eps = o(1)$.} For $\epsilon = o(1)$, Kuszmaul and Walzer \cite{kuszmaul2024space} established a stronger lower bound of 
$$n \log \eps^{-1} + \frac{n}{2} \log \frac{4}{\sqrt{17} - 1} - o(n) \approx n \log \epsilon^{-1} + 0.35n$$ bits. This leaves a gap of approximately $n \log e - n \cdot (\log \frac{4}{\sqrt{17} - 1})/2 \approx 1.1 n$ bits between the best known lower bound \cite{kuszmaul2024space} and the best bound achievable by fingerprint filters.

In this paper, we prove a sharp lower bound of $n \log \eps^{-1} + n \log e - o(n)$ bits for $\eps = o(1)$. The lower bound is information-theoretic, applying regardless of operation time.

\begin{restatable}{theorem}{FilterLB}
\label{thm:filter_lb}
Suppose $U$ is a set and $n$, $\eps$ are parameters such that $\eps = o(1)$ and $|U| = \omega(n \eps^{-1})$. Any algorithm that implements a dynamic filter with universe $U$, capacity $n$, and false-positive error rate $\eps$, and that can support a sequence of $\omega(n)$ insertions and deletions, must use space at least $n \log \eps^{-1} + n \log e - o(n)$ bits.
\end{restatable}

We remark that this theorem is not just of theoretical interest. Practical filters often use an error rate $\eps$ of $1/256$ \cite{lee2021telescoping}, which means $\log \eps^{-1} = 8$. Under these conditions, the $n \log e - o(n) \approx 1.44 n$ term in \cref{thm:filter_lb} accounts for $15.2\%$ of the overall space usage.\footnote{In practice, it is difficult to use a fractional number of bits per key, so real-world implementations tend to incur an overhead of at least $2n$ bits, which accounts for approximately 20\% of the overall space.} Our result demonstrates that this type of overhead is fundamentally unavoidable in both theory and practice.

The proof of \cref{thm:filter_lb} is presented in three parts. We begin in \cref{sec:warmup} by proving a special case of \cref{thm:filter_lb} under two simplifying assumptions: history independence and monotonicity. This special case serves to illustrate our proof framework, in which we construct a one-way communication protocol that is able to extract $n \log \epsilon^{-1} + n \log e - o(n)$ bits from a dynamic filter. Then, in \cref{sec:history_dependent,sec:non-monotone}, we develop techniques to remove these assumptions, first handling history independence and then monotonicity, respectively. The final results apply to any dynamic filter.

%% file: preliminaries.tex
In this section, we formally define the problem that a dynamic filter must solve. Since \cref{thm:filter_lb} focuses on space complexity and does not consider running time, we adopt a simple computational model: the data structure uses a fixed-length memory that can be accessed arbitrarily. For randomized data structures, we additionally assume access to an infinite-length read-only tape of independent random bits, which does not contribute to the space usage.\footnote{Some models assume oracle access to a random number generator that outputs independent random bits on demand. This can be simulated using our random tape model by maintaining a pointer to track the last revealed random bit, allowing us to reveal new random bits as needed. The pointer only occupies $o(n)$ bits of space when there are at most $2^{o(n)}$ calls to the random number generator in total, which is negligible compared to the main space usage of the filter.}

Let $n \in \N$, $\eps \in (0,1)$, and $U = \BK{1,2,\ldots, u}$ be parameters representing the maximum capacity, the false-positive error rate, and the universe, respectively, for the filter. A dynamic filter with respect to these parameters is defined as follows.

\begin{definition}[Dynamic filters]
    A dynamic filter is a randomized data structure $\calD$ that maintains approximate membership for a set $S \subset U$ of at most $n$ elements, supporting the following operations:
    \begin{itemize}
        \item $\texttt{Initialize()}$: Creates and returns an initial empty filter state.
        \item $\texttt{Query(}x\texttt{)}$: For any key $x \in S$, returns $\texttt{true}$ deterministically. For any key $x \notin S$, returns $\texttt{false}$ with probability at least $1 - \eps$, where the probability is taken over the randomness of $\calD$ (the key $x$ is not assumed to be random). 
        \item $\texttt{Insert(}x\texttt{)}$: Adds $x$ to $S$ (i.e., updates $S$ to $S \cup \BK{x}$). This operation assumes $x \notin S$ before insertion.
        \item $\texttt{Delete(}x\texttt{)}$: Removes $x$ from $S$ (i.e., updates $S$ to $S \setminus \BK{x}$). This operation assumes $x \in S$ before deletion.
    \end{itemize}
\end{definition}

We emphasize that, while insertions and deletions update the set $S$, the filter must perform these operations using only the information stored in $\calD$, the input key $x$, and the random tape---the filter does not have explicit knowledge of the current set $S$.

%% file: warmup.tex
In the remainder of this paper, we will prove \cref{thm:filter_lb}. 

\FilterLB*

Fix a filter algorithm $\calD$ that uses $\hfilter$ bits of space. All of the filters in the remainder of this paper will be implemented using $\calD$ with the same random tape. The goal of the proof is to show that $\hfilter \ge n \log \eps^{-1} + n \log e - o(n)$. 

In this section, as a warmup, we first prove the statement for special $\calD$ with the following two assumptions: 
\begin{itemize}
    \item \textbf{History independence:} The memory state of a filter $F$ is uniquely determined by its true key set, along with the random tape used in $\calD$. (Therefore, so is the accepted set $\accept{F}$ of the filter $F$.) Equivalently, any two operational histories that result in the same true key set must also produce the same filter state. 
    \item \textbf{Monotonicity (for history-independent filters):} For any two filters $F$ and $G$ with true key sets $S$ and $T$, respectively, if $T \subseteq S$, then their accepted key sets satisfy $\accept{G} \subseteq \accept{F}$.
\end{itemize}

Formally, in this section, we prove the following proposition:
\begin{proposition}
\label{prop:warmup}
    \cref{thm:filter_lb} holds for filters that satisfy both history independence and monotonicity.
\end{proposition}

\subsection{Proof Intuition for Proposition~\ref{prop:warmup}}
To build intuition, we begin by describing a standard argument that one can use to obtain a weaker lower bound, namely that $\hfilter \ge n\log \eps^{-1} - o(n)$. Assume to the contrary that there exists a filter algorithm $\calD$ with $\hfilter < n \log \eps^{-1} - \Omega(n)$. Such a filter would enable an unrealistically efficient one-way communication protocol for Alice to send a sequence of distinct keys $x_1, \ldots, x_n$ (where each $ x_i \in U$) to Bob.  

On the one hand, from an information-theoretic perspective, Alice must send at least $\log (|U|^{\underline{n}}) \approx n \log |U|$ bits\footnote{The notation $m^{\underline{n}}$ represents the falling factorial of $m$ of order $n$, i.e., $m(m-1)\cdots(m-n+1)$.}, which corresponds to roughly $\log |U|$ bits per key.

On the other hand, if Alice first sends a filter $F$ with a true key set $\BK{x_1, \ldots, x_n}$ and then sends the sequence $(x_1, \ldots, x_n)$ as keys from the accepted key set $\accept{F}$ (which is of size roughly $\eps|U|$), she only needs to send approximately $\log (\eps|U|)$ bits per key---saving $\log \eps^{-1}$ bits per key---along with a filter $F$ of at most $\hfilter < n \log \eps^{-1} - \Omega(n)$ bits. This means that Alice sends a total of $n \log (\eps |U|) + n \log \eps^{-1} - \Omega(n) = n \log |U| - \Omega(n) = \log |U|^{\underline{n}} - \Omega(n)$ bits to indirectly transmit the sequence $(x_1, \ldots, x_n)$ using the filter algorithm $\calD$, which leads to a contradiction, proving the bound.  

Next we explain, at a high level, how to extend this approach to obtain our desired lower bound of $\hfilter \ge n \log \eps^{-1} + n \log e + o(n)$. 
The key observation is that, in the previous protocol of sending $(x_1,\ldots,x_n)$ indirectly, Alice can achieve an even greater reduction in communication cost by using a \emph{dynamic} filter. For instance, after Alice sends $x_1$ using $\log \abs{\accept{F}} \approx \log (\eps |U|)$ bits, Bob can delete $x_1$ from $F$ to obtain $F'$. Alice then only needs to send the remaining keys as elements from $\accept{F'}$. Since $F'$ has a smaller true key set than $F$ and we assume $\calD$ is monotonic, the universe of the remaining keys is reduced by replacing $\accept{F}$ by $\accept{F'}$, leading to an additional reduction in communication cost for Alice. Similarly, after Alice sends $x_2$, Bob can delete $x_2$ from $F'$, further reducing the universe for the remaining keys. More generally, when Alice sends $x_k$, Bob has already deleted $x_1, \ldots, x_{k-1}$ from $F$, resulting in a filter with only ${n-k+1}$ keys in its true key set. 

Ideally, if we assume the size of the accepted set of a filter scales linearly with the size of the true key set---meaning that a filter with a true key set of size $(n-k+1)$ has roughly $\frac{n-k+1}{n}\eps |U|$ accepted keys---then Alice saves an additionally $\log \frac{n}{n-k+1}$ bits when sending $x_k$, which is
\[\sum_{k=1}^{n} \log \frac{n}{n-k+1} = \log \frac{n^n}{n!} \approx n \log e\]
bits in total, improving the previous lower bound to the desired $\hfilter \ge n \log \eps^{-1} + n \log e + o(n)$.

However, in a general filter, the accepted set size may not scale linearly with the true key set size. In particular, the accepted set size might always be very close to $\eps |U|$, meaning that we barely save any communication cost by deleting keys from $F$. In this case, Alice can reduce the message size following another strategy: After Alice sends $x_1$, Bob can build another filter $G$ with true key set $\BK{x_1}$. Then, as all the remaining keys are not in the true key set of $G$, they are unlikely to be in $\accept{G}$, and Alice can send them as elements in $\accept{F} \setminus \accept{G}$ with probability at least $1 - \eps$. This means that in the extreme case where the accepted set size is always almost $\eps|U|$, the universe of all the remaining keys is significantly reduced by replacing $\accept{F}$ with $\accept{F}\setminus \accept{G}$ and so is the communication cost. Similarly, after Alice sends $x_2$, Bob can insert $x_2$ to $G$ and gain further universe reduction.

The filter lower bound in this section cleverly combines the previous two ideas of reducing the universe. Initially, Bob has two filters in his hand, one with the true key set $\BK{x_1, \ldots, x_n}$, called $F$, and the other with an empty true key set, called $G$. 
During the entire protocol, all the keys that are not sent are in the true key set of $F$, but not in the true key set of $G$.
Whenever Alice sends a key $x$, she sends $x$ conditioned on the two filters in Bob's hand (i.e., viewing $x$ as an element of $\accept{F} \setminus \accept{G}$). Then, Bob either deletes $x$ from $F$, or inserts $x$ to $G$. For the convenience of notation, we slightly adjust the order in which Alice sends the keys: In each round, she sends either the smallest unsent key (when Bob is going to insert that key to $G$) or the largest unsent key (when Bob is going to delete that key from $F$) to guarantee that the true key sets of both $F$ and $G$ are prefixes of $\BK{x_1, \ldots,x_n}$. 

In fact, as we will see, the final protocol for Alice actually takes the following simple form: Alice first sends the smallest keys for some number of rounds and then switches to sending the largest keys in the remaining rounds. The main technical challenge in the full proof is to show that, no matter how the accepted set size scales with the true key set size, there always exists such a protocol in which Alice can reduce her overall communication cost by an additional $n\log e - o(n)$ bits when compared to the standard protocol.

\subsection{Formal Proof of Proposition~\ref{prop:warmup}}
Let $(x_1, \ldots, x_n)$ be a sequence of distinct keys from $U$ which is sampled uniformly at random. Consider a one-way communication game where Alice wants to transmit $(x_1, \ldots, x_n)$ to Bob. Suppose both Alice and Bob have free access to the random tape of $\calD$. Note that, no matter how Alice sends her message, she must communicate at least $\log |U|^{\underline{n}}$ bits of entropy.
\begin{claim}
\label{clm:entropy_message_lb_independent}
    In any one-way communication protocol where Alice transmits $(x_1, \ldots, x_n)$ to Bob, the entropy of the message sent by Alice is at least $\log |U|^{\underline{n}}$.
\end{claim}
\begin{proof}
    As Bob can decode $(x_1, \ldots, x_n)$ from the message, the entropy of the message is at least 
    \[H\bk{x_1, \ldots, x_n \mid \text{random tape of $\calD$}} = H\bk{x_1, \ldots, x_n} = \log |U|^{\underline{n}}. \qedhere \] 
\end{proof}

We now construct a protocol where Alice transmits $(x_1, \ldots, x_n)$ indirectly using a filter algorithm. Let $s \in [0, n]$ be an integer parameter to be determined, which will depend only on the filter algorithm $\calD$ but not on the keys $(x_1, \ldots, x_n)$. The protocol $P_s$ (which is parameterized by $s$) is defined in \cref{alg:history_independent_protocol}.
We will also explain \cref{alg:history_independent_protocol} in detail below.

\begin{algorithm}[htbp]
  \caption{Protocol $P_s$}
  \label{alg:history_independent_protocol}
  \SetKwRepeat{Do}{do}{while}
    \SetKwFunction{Send}{Send} 
    \SetKwProg{Fn}{Subroutine}{:}{}
  
  \DontPrintSemicolon
    \Fn(\tcp*[f]{$x$ is in $F$'s true key set, but not in $G$'s.}){\Send{$x, F, G$}}{
        $Z \longleftarrow$ indicator for whether $x \in \accept{G}$\;
        Alice sends $Z$ to Bob\;
        \uIf {$Z = 0$ (i.e., $x \notin \accept{G}$)} {
            Alice sends $x$ as an element in $\accept{F} \setminus \accept{G}$\;
        } \Else{
            Alice sends $x$ as an element in $U$\;
        }  
    }
    $F \longleftarrow$ filter with true key set $\BK{x_1, \ldots, x_n}$\;
    $G \longleftarrow$ filter with true key set $\varnothing$\;
    Alice prepares $F$ and sends it to Bob\;
    Bob prepares $G$ by himself\;
    \For{$k$ from $1$ to $s$} { 
        Alice sends $x_k$ using \Send{$x_k, F, G$}\;
        Bob inserts $x_k$ to $G$\;
    }
    \For{$k$ from $n$ to $s + 1$} {
        Alice sends $x_k$ using \Send{$x_k, F, G$}\;
        Bob deletes $x_k$ from $F$\;
    }
\end{algorithm}

\paragraph{Initialization.}
At the beginning of the protocol $P_s$, Alice first sends a filter $F$ with the true key set $\BK{x_1, \ldots, x_n}$ to Bob, and Bob prepares a filter $G$ with an empty true key set. This step takes $\hfilter$ bits of communication.

\paragraph{Sending a single key.}
In the main part of the protocol, Alice sends each key $x$ to Bob using the protocol $\send{x, F, G}$, where $F$ and $G$ are two filters held by Bob such that $x$ is in the true key set of $F$ but not in the true key set of $G$.

In the protocol $\send{x, F, G}$, Alice first sends a bit $Z$ to indicate whether $x$ is in the accepted set of $G$. As $x$ is not in the true key set of $G$, $Z$ is a highly biased bit, with probability at most $\eps$ that $Z$ equals $1$. Hence, the entropy of $Z$ is small: 
\begin{align*}
    H(Z) \le h(\eps) \defeq -\eps\log \eps - (1-\eps) \log (1 - \eps),
\end{align*}
where $h$ denotes the binary entropy function.

Then, depending on the different values of $Z$, Alice sends $x$ differently. In the common case where $x \notin \accept{G}$, $x$ must lie in $\accept{F} \setminus \accept{G}$ (since $x$ is in the true key set of $F$, it is also in the accepted set $\accept{F}$). Alice can then send $x$ as an element of $\accept{F} \setminus \accept{G}$ using only $\log \abs{\accept{F} \setminus \accept{G}}$ bits. In the rare case that $x \in \accept{{G}}$, Alice sends $x$ directly as an element in $U$ using $\log |U|$ bits. The message length for this step can be written as 
\begin{align*}
    Z \log |U| + (1-Z) \log \abs{\accept{F} \setminus \accept{G}},
\end{align*}
where both $Z$ and $\abs{\accept{F} \setminus \accept{G}}$ are random variables.

After Alice sends the key $x$ using $\send{x, F, G}$, Bob will either insert $x$ to $G$ or delete $x$ from $F$ to reduce the communication cost for future keys.

\paragraph{Order of key sendings.} The protocol $P_s$ has two stages. In the first stage, Alice always sends the smallest unsent key in each round (i.e., in the order of $x_1, x_2, \ldots$), and after receiving it, Bob inserts the key into $G$. Once Alice has sent $x_1, \ldots, x_s$, they switch to the second stage where Alice instead sends the largest unsent key in each round (i.e., in the order of $x_n, x_{n-1}, \ldots$) and Bob deletes the key from $F$. 

Throughout the protocol, the true key sets of both $F$ and $G$ are prefixes of $\BK{x_1, \ldots, x_n}$. Specifically, for any $k \le n$, let $F_k$ denote the filter with the true key set $\BK{x_1, \ldots , x_k}$. Since we assume $\calD$ is history independent, the true key set uniquely determines the filter. Then, when Alice sends $x_k$, the two filters $F$ and $G$ held by Bob can be represented as $F_{r_k}$ and $F_{\l_k}$ with $\l_k < k \le r_k$, where 
\begin{align*}
    \l_k \defeq \begin{cases}
        k - 1 & \textup{if } k \le s\\
        s & \textup{if } s+ 1 \le k \le n
    \end{cases} \quad  \text{and} \quad
    r_k \defeq \begin{cases}
        n & \textup{if }k \le s \\
        k& \textup{if } s+1 \le k \le n
    \end{cases}.
\end{align*}
Plugging this into the previous calculation, when Alice sends $x_k$ using $\texttt{Send}(x_k, F_{r_k}, F_{\l_k})$, she first sends an indicator $Z_k \defeq \ind\Bk{x_k \in \accept{F_{\l_k}}}$, and then sends $Z_k \log |U| + (1-Z_k) \log \abs{\accept{F_{r_k}} \setminus \accept{F_{\l_k}}}$ bits. The total entropy of the message for sending $x_k$ is at most
\begin{align*}
    h(\eps) + \E\Bk*{Z_k \log |U| + (1-Z_k) \log \abs*{\accept{F_{r_k}} \setminus \accept{F_{\l_k}}}},
\end{align*}
where the expectation is taken over both the randomness of $\calD$ and the randomness of the keys $\bk{x_1, \ldots, x_n}$.

\paragraph{Upper bounding the total message entropy.}
Summing up the message entropy for sending the filter $F_n$ and the message entropy for sending each $x_k$, the total message entropy of $P_s$ is at most
\begin{align}
    \hfilter + n h(\eps) + \sum_{k=1}^n \E\Bk*{Z_k \log |U| + (1-Z_k) \log \abs*{\accept{F_{r_k}} \setminus \accept{F_{\l_k}}}}. \label{eq:entropy_message_independent}
\end{align}
Below, we further simplify \eqref{eq:entropy_message_independent}.

Define $a_k$ for $k = 0,1, \ldots, n$ as 
\[a_k \defeq \frac{1}{(1 - \eps) n + \eps|U|} \E\Bk*{\abs{\accept{F_k}} - \abs{\accept{F_{k-1}}}}, \]
where we adopt the convention that $|\accept{F_{-1}}| = 0$. Since we assume $\calD$ is monotonic, the accepted sets satisfy $\accept{F_0} \subseteq \accept{F_1} \subseteq \cdots \subseteq \accept{F_n}$, so each $a_k$ is non-negative. 
Note that $a_k$ is not a random variable and depends only on the algorithm $\calD$ since we already take expectation over all the randomness (including the randomness of $\calD$ and the randomness of the keys $\bk{x_1, \ldots, x_n}$) in the definition of $a_k$. Below, we write \eqref{eq:entropy_message_independent} in terms of $a_k$'s.

First, by the definition of $a_k$, the expected size of the accepted sets can be written as
\begin{align*}
    \E\Bk*{\abs{\accept{F_k}}} = \bk*{(1 - \eps) n + \eps |U|} \sum_{i=0}^k a_i = \bk*{(1 - \eps) n + \eps |U|} \cdot a_{[0,k]}.
\end{align*}
Here, for clarity of notation, we introduce two abbreviations for partial sums of $a_k$: For any $\l \le r$, we use $a_{[\l,r]}$ to denote $\sum_{i=\l}^r a_{i}$ and we use $a_{(\l, r]}$ to denote $\sum_{i = \l+1}^r a_i$.
Moreover, since a key outside the true key set of a filter falls into the accepted set with probability at most $\eps$, we obtain $\E\Bk*{\abs{\accept{F_n}}} \le n + \eps \bk{|U| - n} = (1-\eps) n + \eps |U|$, which means $a_{[0,n]} = a_0 + \cdots + a_n \le 1$.

Then, each term in the summation of \eqref{eq:entropy_message_independent} can be upper bounded in terms of $a_k$'s by the following claim:
\begin{claim}
    \label{clm:entropy_single_key_independent}
    For any $k \le n$,
    \begin{align*}
        \E\Bk*{Z_k \log |U| + (1-Z_k) \log \abs*{\accept{F_{r_k}} \setminus \accept{F_{\l_k}}}}
        {}\le{}  \log |U| + (1 - \eps) \log a_{(\l_k, r_k]} + \log \eps + o(1).
    \end{align*}
\end{claim}
\begin{proof}
    Recall that $Z_k \defeq \ind\Bk{x_k \in \accept{F_{\l_k}}}$ and $p \defeq \Pr\Bk{Z_k = 1} \le \eps$. By the principle of total probability, 
    \begin{align*}
        &\E\Bk*{Z_k \log |U| + (1-Z_k) \log \abs*{\accept{F_{r_k}} \setminus \accept{F_{\l_k}}}}\\
        {}={}& p\log |U| + (1 - p)\E\Bk*{\log \bk*{\abs*{\accept{F_{r_k}}} - \abs{\accept{F_{\l_k}}}} \mid Z_k = 0}\\
        {}\le{}& p\log |U| + (1 - p)\log \bk*{\E\Bk*{\abs*{\accept{F_{r_k}}} - \abs{\accept{F_{\l_k}}}\mid Z_k = 0}}, \numberthis \label{ineq:entropy_single_key_independent}
    \end{align*}
    where the second inequality is because the function $\log$ is concave.
    Moreover, 
    \begin{align*}
        &\E\Bk*{\abs*{\accept{F_{r_k}}} - \abs{\accept{F_{\l_k}}}\mid Z_k = 0}\\
        {}={}& \sum_{i=0}^{\infty} i \cdot \Pr\Bk*{\abs*{\accept{F_{r_k}}} - \abs*{\accept{F_{\l_k}}} = i \mid Z_k = 0}\\
        {}\le{}& \sum_{i=0}^{\infty} i \cdot \frac{\Pr\Bk*{\abs*{\accept{F_{r_k}}} - \abs*{\accept{F_{\l_k}}} = i}}{\Pr\Bk*{Z_k = 0}}\\
        {}={}&  \frac{\E\Bk*{\abs*{\accept{F_{r_k}}} - \abs*{\accept{F_{\l_k}}}}}{\Pr\Bk*{Z_k = 0}}
        {}\le{}  \frac{{(1 - \eps)n + \eps |U|} }{1- \eps} \cdot a_{(\l_k, r_k]}, \numberthis \label{ineq:conditional_accepted_set_size_independent}
    \end{align*}
    and it is clear that $\E\Bk*{\abs*{\accept{F_{r_k}}} - \abs{\accept{F_{\l_k}}}\mid Z_k = 0} \le |U|$, which implies that \eqref{ineq:entropy_single_key_independent} is increasing in $p$. Therefore, by plugging \eqref{ineq:conditional_accepted_set_size_independent} and $p \le \eps$ into \eqref{ineq:entropy_single_key_independent}, we obtain the desired bound
    \begin{align*}
        &\E\Bk*{Z_k \log |U| + (1-Z_k) \log \abs*{\accept{F_{r_k}} \setminus \accept{F_{\l_k}}}}\\
        {}\le{}& p\log |U| + (1 - p)\log \bk*{\E\Bk*{\abs*{\accept{F_{r_k}}} - \abs{\accept{F_{\l_k}}}\mid Z_k = 0}}\\
        {}\le{}& \eps\log |U| + (1 - \eps)\log \bk*{ \frac{{(1 - \eps)n + \eps |U|} }{1- \eps} \cdot a_{(\l_k, r_k]}}\\
        {}\le{}& \eps\log |U| + (1 - \eps)\log \bk*{ \eps |U| \cdot a_{(\l_k, r_k]}} + o(1)\\
        {}\le{}& \log |U| + (1 - \eps)\log a_{(\l_k, r_k]} + \log \eps + o(1),
    \end{align*}
    where the third and fourth inequalities are because $\eps = o(1)$ and $|U| = \omega(n \eps^{-1})$, which imply $\log (1 - \eps) = o(1)$, $\log\bk*{\frac{(1-\eps)n + \eps |U|}{\eps |U|}} = o(1)$, and $\eps\log \eps = o(1)$.
\end{proof}

Now, plugging \cref{clm:entropy_single_key_independent} into \eqref{eq:entropy_message_independent}, the total message entropy of $P_s$ is at most
\begin{align*}
    &\hfilter + nh(\eps) + \sum_{k=1}^n \bk*{\log |U| + (1 - \eps) \log a_{(\l_k, r_k]} + \log \eps + o(1)}\\
    {}\le{}&\hfilter + n\log |U| + n \log \eps + (1 - \eps)\sum_{k=1}^n  \log a_{(\l_k, r_k]} + o(n)\\
    {}={}&\hfilter + n\log |U| + n \log \eps + (1 - \eps)\bk*{\sum_{k=1}^s  \log a_{[k, n]} + \sum_{k=s+1}^n \log a_{(s,k]}} + o(n), \label{ineq:entropy_message_in_ak_independent} \numberthis\footnotemark
\end{align*}
\footnotetext{When $s = 0$, the sum $\sum_{k=1}^s \log a_{[k,n]}$ is empty and vanishes; when $s = n$, the sum $\sum_{k = s+1}^n \log a_{(s, k]}$ is empty and vanishes.}
where we use $h(\eps) = o(1)$ in the first inequality.

\paragraph{Choice of \texorpdfstring{$s$}{s}.} Now, comparing \eqref{ineq:entropy_message_in_ak_independent} with the message entropy lower bound in \cref{clm:entropy_message_lb_independent}, we obtain
\begin{align*}
    \hfilter 
    {}\ge{} & \log |U|^{\underline{n}} - n\log |U| - n\log \eps - (1-\eps)\bk*{\sum_{k=1}^s  \log a_{[k, n]} + \sum_{k=s+1}^n \log a_{(s,k]}} - o(n)\\
    {}\ge{} & n\log \eps^{-1} - (1-\eps)\bk*{\sum_{k=1}^s  \log a_{[k, n]} + \sum_{k=s+1}^n \log a_{(s,k]}} - o(n),\label{ineq:filter_size_lb_independent} \numberthis
\end{align*}
where we use $\log |U|^{\underline{n}} - n \log |U| \ge n \log \bk*{1 - \frac{n}{|U|}} = -o(n)$ as $|U| = \omega(n)$. 

Thus, to establish the desired lower bound for $\hfilter$, it suffices to show that there exists a choice of $s$ for which 
$${\sum_{k=1}^s  \log a_{[k, n]} + \sum_{k=s+1}^n \log a_{(s,k]}} \le - n \log e + o(n),$$
or equivalently that
$${\sum_{k=1}^s  \log \frac{1}{a_{[k, n]}} + \sum_{k=s+1}^n \log \frac{1}{a_{(s,k]}}} \ge n \log e - o(n).$$
The next lemma establishes that such an $s$ always exists, and is the main technical lemma of the section. 

\begin{lemma}
    \label{lem:choice_of_s}
    Let $a_1, \ldots, a_n \ge 0$ be real numbers such that $a_1 + \cdots + a_n \le 1$. Then, there exists an $s$ with $0 \le s \le n$, such that 
    \begin{align*}
        \sum_{k=1}^s  \log a_{[k, n]} + \sum_{k=s+1}^n \log a_{(s,k]} \le -n\log e + o(n). \label{ineq:choice_of_s} \numberthis
    \end{align*}
\end{lemma}

Using \cref{lem:choice_of_s}, we can choose a suitable $s$ such that \eqref{ineq:filter_size_lb_independent} becomes
\begin{align*}
    \hfilter 
    {}\ge{}  n\log \eps^{-1} + (1-\eps)n\log e - o(n)
    {}\ge{}  n\log \eps^{-1} + n\log e - o(n).
\end{align*}
Furthermore, the protocol $P_s$ requires only $O(n)$ insertions and deletions to maintain the filters $F$ and $G$. Therefore, assuming \cref{lem:choice_of_s}, we have completed the proof of \cref{prop:warmup}.

\subsection{Proof of Lemma \ref{lem:choice_of_s}}
In this subsection, we prove \cref{lem:choice_of_s}. First, by the monotonicity of the logarithm function, we may assume without loss of generality that $a_1 + \cdots + a_n = 1$. Moreover, instead of directly finding an $s$ that satisfies \eqref{ineq:choice_of_s}, we will inductively establish the existence of an $s$ that meets a more precise constraint:
\begin{align*}
        \sum_{k=1}^s  \log a_{[k, n]} + \sum_{k=s+1}^n \log a_{(s,k]} \le \log \frac{n!}{n^n}. \label{ineq:choice_of_s_stronger}\numberthis
\end{align*}
By Stirling's approximation, \eqref{ineq:choice_of_s_stronger} implies \eqref{ineq:choice_of_s}, since $\log \frac{n!}{n^n} \le \log \frac{(n/e)^n}{n^n} + O(\log n) = -n\log e + o(n)$.

\begin{proofof}{\cref{lem:choice_of_s}}
    We use induction on $n$ to prove the existence of $s$ satisfying \eqref{ineq:choice_of_s_stronger}. 

    For the base case $n = 1$, we can choose either $s = 0$ or $1$, and \eqref{ineq:choice_of_s_stronger} simplifies to $\log a_1 \le 0$, which follows from the condition $a_1 = 1$. For the rest of the proof, consider some $n > 1$, and suppose that the statement holds for all integers smaller than $n$.
    
    First, we check whether $s = 0$ already satisfies \eqref{ineq:choice_of_s_stronger}. If
    \begin{align*}
      \sum_{k=1}^n \log a_{(0,k]} \le \log \frac{n!}{n^n}, \label{ineq:choice_of_s_s=0} \numberthis
    \end{align*}
    then we can directly take $s = 0$ to conclude the statement. Below, we assume \eqref{ineq:choice_of_s_s=0} does not hold.
    
    Now, let $j$ be the smallest index such that 
    \begin{align*}
      \sum_{k=1}^j \log a_{(0,k]} > \sum_{k = 1}^j \log \frac{k}{n}. \label{ineq:definition_of_j} \numberthis
    \end{align*}
    The existence of such $j$ is ensured by the assumption that \eqref{ineq:choice_of_s_s=0} does not hold, implying that $j = n$ is always a valid candidate. By the definition of $j$, we have
    \begin{align*}
        \sum_{k=1}^i \log a_{(0,k]} \le \sum_{k=1}^i \log \frac{k}{n}, \quad \forall i < j. \label{ineq:property_of_j_majorization} \numberthis
    \end{align*}
    In particular, by comparing \eqref{ineq:definition_of_j} and \eqref{ineq:property_of_j_majorization} (taking $i = j-1$), we get that $\log a_{(0, j]} > \log (j/n)$, which implies that 
    \begin{equation}a_1 + \cdots + a_j = a_{(0, j]} > j/n
    \label{eq:a0j}
    \end{equation}

    We will choose $s$ from the range $j \le s \le n$ by using the inductive hypothesis
    on a normalized suffix 
    \begin{align*}
        \bk*{\frac{a_{j+1}}{a_{j+1} + \cdots + a_n}, \frac{a_{j+2}}{a_{j+1} + \cdots + a_n}, \ldots, \frac{a_{n}}{a_{j+1} + \cdots + a_n}}.
    \end{align*}
    By induction hypothesis, there exists $s$ with $j \le s \le n$, such that 
    \begin{align*}
        \sum_{k = j+1}^s \log a_{[k,n]} + \sum_{k=s+1}^n \log a_{(s,k]} \le \log \frac{(n-j)!}{(n-j)^{n-j}} + (n - j) \log a_{(j,n]}.\label{ineq:using_induction_hypothesis} \numberthis
    \end{align*}
    Below, we show that such an $s$ also satisfies \eqref{ineq:choice_of_s_stronger}.

    Using \eqref{ineq:using_induction_hypothesis}, the left-hand side of \eqref{ineq:choice_of_s_stronger} can be upper bounded by
    \begin{align*}
        \sum_{k=1}^s  \log a_{[k, n]} + \sum_{k=s+1}^n \log a_{(s,k]}
        {}\le{}& \sum_{k=1}^j  \log a_{[k, n]} + \log \frac{(n-j)!}{(n-j)^{n-j}} + (n - j) \log a_{(j,n]}\\
        {}={}& \sum_{k=2}^{j+1}  \log a_{[k, n]} + \log \frac{(n-j)!}{(n-j)^{n-j}} + (n - j - 1) \log a_{(j,n]}, \label{ineq:after_plugging_in_hypothesis} \numberthis
    \end{align*}
    where in the second line we use $\log a_{[1,n]} = 0$ and $\log a_{[j+1, n]} = \log a_{(j, n]}$. The main step in the rest of the proof will be to establish an upper bound on $ \sum_{k=2}^{j+1}  \log a_{[k, n]}$ using \eqref{ineq:definition_of_j} and \eqref{ineq:property_of_j_majorization}.

    Let $f: \R_{+} \to \R$ denote the function $f(x) = \log \bk*{1 - 2^{-x}}$. Then, one can observe that
    \begin{align*}
        \sum_{k=2}^{j+1}  \log a_{[k, n]}
        = \sum_{k=1}^{j}  \log a_{(k, n]}
        = \sum_{k=1}^{j}  \log \bk*{1 - a_{(0,k]}}
        = \sum_{k=1}^{j}  f\bk*{-\log a_{(0,k]}}. \label{ineq:sum_of_convex_function} \numberthis
    \end{align*}
    The right-hand side of \eqref{ineq:sum_of_convex_function} can be upper bounded using a variation of Karamata's inequality \cite{karamata1932inegalite}:
    \begin{theorem}[Karamata's inequality \cite{karamata1932inegalite}] 
        Let $x_1 \ge x_2 \ge\cdots \ge x_n$ and $y_1 \ge y_2 \ge \cdots \ge y_n$ be two sequences of real numbers, satisfying
        \begin{itemize}
            \item $x_1 + \cdots + x_i \ge y_1 + \cdots + y_i$, for all $i < n$, and
            \item $x_1 + \cdots + x_n \le y_1 + \cdots + y_n$.
        \end{itemize}
        Then, for any concave and increasing function $f: \R \to \R$, we have 
        \begin{align*}
            f(x_1) + \cdots + f(x_n) \le f(y_1) + \cdots + f(y_n).
        \end{align*}
    \end{theorem}
    For an overview of Karamata's inequality, and other related inequalities, see also \cite{kadelburg2005inequalities}.
    
    We will use Karamata's inequality over the function $f$ and two sequences $\bk*{-\log a_{(0,1]}, -\log a_{(0,2]}, \ldots, -\log a_{(0,j]}}$ and $\bk*{-\log (1/n), -\log (2/n), \ldots, -\log (j/n)}$. Clearly, both sequences are non-increasing. Moreover, the premise of Karamata's inequality is satisfied by these two sequences according to \eqref{ineq:definition_of_j} and \eqref{ineq:property_of_j_majorization}. We can also check $f$ is concave and increasing as follows.
    \begin{claim}
        The function $f: \R_+ \to \R$ defined as $f(x) = \log \bk*{1 - 2^{-x}}$ is concave and increasing.
    \end{claim}
    \begin{proof}
        $f'(x) = \frac{2^{-x}}{1 - 2^{-x}} > 0 $ for all $x > 0$, and $f''(x) = \frac{-2^{x} \ln 2}{(2^x - 1)^2} < 0$ for all $x > 0$.
    \end{proof}

    Hence, we can use Karamata's inequality over the function $f$ and the two sequences to obtain
    \begin{align*}
        \sum_{k=1}^j f\bk*{-\log a_{(0,k]}} 
        \le \sum_{k=1}^j f\bk*{-\log \frac{k}{n}} 
        = \sum_{k=1}^j \log \frac{n-k}{n}
        = \sum_{k=n-j}^{n-1} \log \frac{k}{n}. \label{ineq:applying_karamata} \numberthis
    \end{align*}
    Finally, by plugging \eqref{ineq:sum_of_convex_function} and \eqref{ineq:applying_karamata} into \eqref{ineq:after_plugging_in_hypothesis}, we obtain
    \begin{align*}
        & \sum_{k=1}^s  \log a_{[k, n]} + \sum_{k=s+1}^n \log a_{(s,k]}\\
        {}\le{}& \sum_{k=2}^{j+1}  \log a_{[k, n]} + \log \frac{(n-j)!}{(n-j)^{n-j}} + (n - j - 1) \log a_{(j,n]} \tag{by \eqref{ineq:after_plugging_in_hypothesis}}\\
        {}\le{}& \sum_{k=n-j}^{n-1} \log \frac{k}{n} + \log \frac{(n-j)!}{(n-j)^{n-j}} + (n - j - 1) \log a_{(j,n]} \tag{by \eqref{ineq:sum_of_convex_function} and \eqref{ineq:applying_karamata}}\\
        {}\le{}& \sum_{k=n-j}^{n-1} \log \frac{k}{n} + \log \frac{(n-j)!}{(n-j)^{n-j}} + (n - j - 1) \log \frac{n-j}{n},
        \end{align*}
        where in the final step uses that $a_{(j,n]} = 1 - a_{(0,j]} \le 1 - j/n$ by \eqref{eq:a0j}. Continuing, the bound simplifies to
        \begin{align*}
        {}{}& \sum_{k=n-j}^{n-1} \log \frac{k}{n} + \sum_{k=1}^{n-j-1}\bk*{\log \frac{k}{n-j} + \log \frac{n-j}{n}}
        {}={} \sum_{k=1}^{n} \log \frac{k}{n} = \log \frac{n!}{n^n},
    \end{align*}
    as desired. This completes the induction step and proves \cref{lem:choice_of_s}.
\end{proofof}

%% file: history_dependent.tex
In this section, we extend the lower bound in \cref{sec:warmup} to remove the assumption of history independence. For a history-dependent filter, the memory state is not purely determined by the true key set and the random seeds, but also depends on the history of operations used to obtain this filter.

Note that we still assume the data structure $\calD$ is monotone, with a generalized definition of monotonicity defined as follows (the previous monotonicity in \cref{sec:warmup} is for history-independent filters):

\begin{definition}[Self-contained operational sequence]
    We say an operational sequence $\sigma$ consisting of a series of insertions and deletions is \defn{self-contained} if it satisfies the following constraints:
    \begin{itemize}
        \item (No deletions of non-elements.) Each deleted key in $\sigma$ must have been inserted earlier \emph{in the same sequence} $\sigma$ and has not been deleted in between.
        \item (No duplicate insertions.) For any key that appears multiple times in $\sigma$, between any two insertions of that key, there must be a deletion of that key.
    \end{itemize}
    Equivalently, a self-contained operational sequence is a valid sequence of operations that can be performed on an empty filter.
\end{definition}

Clearly, the true key set of a filter is non-decreasing after we apply a self-contained operational sequence over the filter.

\begin{definition}[Monotonicity for general filters]
    We say a filter algorithm $\calD$ is \defn{monotone} if the following is true: For any $0 < j < k$, and for any sequence $\sigma = (\sigma(1), \sigma(2), \ldots, \sigma(k))$ of operations starting with an empty filter, and where the suffix $\sigma(j + 1), \ldots, \sigma(k)$ is self-contained, the filter $G$ obtained after operation $\sigma(j)$ and the filter $F$ obtained after operation $\sigma(k)$ are guaranteed to satisfy $\accept{G} \subseteq \accept{F}$.
\end{definition}

Formally, in this section, we prove the following proposition:

\begin{proposition}
\label{prop:history_dependent}
    \cref{thm:filter_lb} holds for monotone filters.
\end{proposition}

\subsection{Proof Intuition for Proposition~\ref{prop:history_dependent}}
\label{subsec:intuition_dependent}
We begin by analyzing the proof of \cref{prop:warmup} to identify precisely where the assumptions of history independence and monotonicity are used.

In the proof of \cref{prop:warmup}, Bob initially has two filters $F$ and $G$ with true key sets $\BK{x_1, \ldots, x_n}$ and $\varnothing$, respectively. Whenever Alice sends a key, Bob either deletes it from $F$ or inserts it into $G$, keeping both the true key sets of $F$ and $G$ being prefixes of $\BK{x_1, \ldots, x_n}$. 
For any $k \in \BK{0,1,\ldots, n}$, let $F_k$ denote the filter obtained by deleting $x_n, x_{n-1}, \ldots, x_{k+1}$ from the initial $F$ and let $G_k$ denote the filter obtained by inserting $x_1, x_2, \ldots, x_{k}$ into the initial $G$. In \cref{prop:warmup}, where we assume $\calD$ is history independent, since $F_k$ and $G_k$ share the same true key set $\BK{x_1, \ldots, x_k}$, they are actually the same filter states. However, in \cref{prop:history_dependent} (without history independence), as $F_k$ (obtained by first inserting $x_1, \ldots, x_n$ to an empty filter and then deleting $x_n, \ldots, x_{k+1}$) and $G_k$ (obtained by directly inserting $x_1, \ldots, x_k$ to an empty filter) are constructed from different operational histories, they are different filter states.

The proof of \cref{prop:warmup} heavily relies on the fact that $F_k$ and $G_k$ represent the same filter state to upper bound the communication cost. When Alice sends $x_k$ to Bob, the expected communication cost is primarily determined by how many elements are in the sub-universe $\accept{F_{r_k}} \setminus \accept{G_{\l_k}}$. In particular, the expected communication cost is approximately  
\[
    \E\Bk*{\log \abs*{\accept{F_{r_k}} \setminus \accept{G_{\l_k}}}} 
    \le \log \E\Bk*{\abs*{\accept{F_{r_k}} \setminus \accept{G_{\l_k}}}}
    = \log \bk*{\E\Bk*{\abs*{\accept{F_{r_k}}}} - \E\Bk*{\abs*{\accept{G_{\l_k}}}}},
\]
where the equality follows from monotonicity. Since $F_k$ and $G_k$ are identical filter states, we can substitute $\E\Bk*{\abs*{\accept{G_{\l_k}}}}$ with $\E\Bk*{\abs*{\accept{F_{\l_k}}}}$ in the above equation. This substitution enables us to apply \cref{lem:choice_of_s} to select a suitable parameter $s$ that the protocol can use to achieve a significant reduction in communication cost. However, the same argument does not extend to the case where $F_k$ and $G_k$ are generally different filter states. Specifically, it is possible that for each $k \leq n$, the size $\abs{\accept{F_k}}$ remains close to $\eps |U|$, while $\abs*{\accept{G_k}}$ stays close to zero. In this scenario, no choice of $s$ would significantly reduce the communication cost. 

To address this issue, the main technical contribution of this section is to modify the definitions of $F_k$ and $G_k$ using a more complicated operational sequence (called the \defn{obfuscation sequence}). This modification ensures that, although $F_k$ and $G_k$ are still distinct filter states, their distributions under all sources of randomness are close. This closeness guarantees that $\E\Bk*{\abs*{\accept{F_{k}}}}$ is close to $\E\Bk*{\abs*{\accept{G_k}}}$, which in turn allows us to complete the proof.

To illustrate the high-level idea of the obfuscation sequence, we start with the filter states $F_k$'s and $G_k$'s defined previously and demonstrate how to modify them step by step to make their distributions close. Since $F_n$ and $G_n$ are already identical (both obtained by inserting $x_1, \ldots, x_n$), the first nontrivial case to consider is obfuscating $F_{n-1}$ and $G_{n-1}$. 
Recall that, compared to the operational sequence $\sigma_{G_{n-1}}$ that produces $G_{n-1}$ (i.e., inserting $x_1, \ldots, x_{n-1}$), the sequence $\sigma_{F_{n-1}}$ leading to $F_{n-1}$ includes an additional insertion and deletion of $x_n$, making the two filter states different. To obfuscate this difference, we introduce additional operations in $\sigma_{G_{n-1}}$: After inserting $x_1, \ldots, x_{n-1}$, we repeatedly insert and delete a random key, performing this operation a random number of times chosen from $[0, M - 1]$ for a large enough parameter $M$. Then, due to the randomness of $x_n$, the operational sequence $\sigma_{F_{n-1}}$ can be interpreted as follows: after inserting $x_1, \ldots, x_{n-1}$, we again repeatedly insert and delete a random key, but now for a random number of times chosen from $[1, M]$. Since these distributions are close, this brings the distributions of $F_{n-1}$ and $G_{n-1}$ closer together. 

Here, we emphasize that it is crucial to maintain $\sigma_{G_{n-1}}$ as a prefix of $\sigma_{F_{n-1}}$, as this allows us to leverage the monotonicity between $F_{n-1}$ and $G_{n-1}$. This requirement rules out naive obfuscation strategies, such as keeping $\sigma_{F_{n-1}}$ unchanged while simply appending an additional insertion and deletion of a random key at the end of $\sigma_{G_{n-1}}$. Such an approach would disrupt the structural relationship between the sequences, making it ineffective for our purpose.

Similarly, we can further obfuscate $F_{n-2}$ and $G_{n-2}$ using the same strategy. Suppose the sequence $\sigma_{F_{n-2}}$ leading to $F_{n-2}$ is longer than the sequence $\sigma_{G_{n-2}}$ leading to $G_{n-2}$ by a suffix $\sigma'$. To obfuscate this difference, we introduce additional operations at the end of $\sigma_{G_{n-2}}$ as follows: Sample independent sequences $\sigma_1', \ldots, \sigma_i'$ from the same distribution as $\sigma'$, where $i$ is a random number drawn from $[0, M-1]$. Then, after inserting $x_1, \ldots, x_{n-2}$, we append $\sigma_1', \ldots, \sigma_i'$ to the end of $\sigma_{G_{n-2}}$, making its distribution closely match that of $\sigma_{F_{n-2}}$.  

This process can be applied recursively to obfuscate $F_{n-3}$ and $G_{n-3}$, $F_{n-4}$ and $G_{n-4}$, continuing until we reach $F_0$ and $G_0$. To formally define the resulting operational sequences and capture the recursive structure, we introduce a tree-based framework called the \emph{obfuscating tree}, which will be detailed in the full proof.  

Finally, even with these obfuscated filter states, several technical challenges remain. The previous construction of operational sequences involves $n$ recursive steps, each increasing the sequence length by a factor of $M$. As a result, the final sequence grows exponentially, which is impractical.  

To mitigate this, Alice divides the $n$ keys $x_1, \dots, x_n$ into $b$ batches, where $b$ is a relatively small superconstant, with each batch containing $n/b$ keys. Bob updates his filter only after receiving an entire batch, reducing the number of filters he maintains to $b$. This modification limits the recursion depth to $b$, yielding an operational sequence of length $M^b$ instead of $M^n$.  

Additionally, for technical reasons, Alice and Bob first agree on a random partition of the universe $U = U_1 \cup U_2 \cup \dots \cup U_b$, ensuring that keys in the $k$-th batch are sampled from $U_k$. While unnecessary in this section, this step will be crucial when addressing the removal of the monotonicity assumption in \cref{sec:non-monotone}.

\subsection{Formal Proof of Proposition~\ref{prop:history_dependent}}
The proof of \cref{prop:history_dependent} is still based on a one-way communication game between Alice and Bob. Suppose Alice and Bob have access to the random tape of $\calD$ together with another public random tape with free randomness. At the beginning of the game, Alice and Bob first agree on a random partition of the universe that is determined by the free public randomness.

\paragraph{Random partition of the universe.} Let $b = \omega(1)$ be a parameter to be determined. Let $\pi = \bk{U_1, \ldots, U_b}$ be a partition of the universe $U$ that is sampled uniformly at random, such that $|U_k| = |U|/b$ for each $k \le b$.
Based on this random partition $\pi$, the goal of the communication game is as follows.

For each $k \le b$, let $X_k = \bk*{x^{(k)}_1 , \ldots, x^{(k)}_{n/b}}$ be a sequence of distinct keys from $U_k$ which is sampled uniformly at random. Alice wants to transmit $\bk*{X_1, \ldots, X_b}$ to Bob via the one-way communication protocol. Similar to \cref{clm:entropy_message_lb_independent}, no matter how Alice sends her message, she needs to communicate at least $n\log\bk{|U|/b} - o(n)$ entropy.

\begin{claim}
    \label{clm:message_entropy_lb_dependent}
    The message entropy sent by Alice is at least $b \log \bk*{\bk*{{|U|}/{b}}^{\underline{n/b}}}$.\footnote{As in earlier sections, we use the notation $m^{\underline{n}}$ to represent the falling factorial of $m$ of order $n$, i.e., $m(m-1)\cdots(m-n+1)$.}
\end{claim}
\begin{proof}
    As Bob can decode $(X_1, \ldots, X_b)$ from the message, the entropy of the message is at least
    \begin{align*}
        &H\bk{X_1, \ldots, X_b \mid \text{random tapes}}
        = H\bk{X_1, \ldots, X_b \mid \pi}
        = \sum_{k=1}^b H\bk{X_k\mid U_k} 
        = b \log \bk*{\bk*{{|U|}/{b}}^{\underline{n/b}}}. \qedhere
    \end{align*}
\end{proof}

\paragraph{Obfuscation sequences.}

Next, we construct the protocols $Q_s$ that allow Alice to send $(X_1, \ldots, X_b)$ indirectly to Bob. The protocol is similar to the protocols $P_s$, but we will perform additional obfuscation operations on the filters $F$ and $G$ to make their distributions similar, as described in \cref{subsec:intuition_dependent}. Below, we construct the sequence of operations that we will use for constructing $F$ and $G$---we call this the \defn{obfuscation sequence}.

We first construct a random depth-$b$ tree $T$, referred to as the \defn{obfuscating tree}. Let $M$ be a large integer parameter to be determined. The tree $T$ is built recursively: We begin with the root node $v_0$, designated as the \emph{level-0 node}. Using the public random tape, we sample $d_{v_0}$ as a uniformly random integer from $[1, M]$ and create $d_{v_0}$ children for $v_0$, forming the \emph{level-1 nodes}. 
Similarly, for each level-1 node $v$, we independently sample $d_v \in [1, M]$ and create $d_v$ children. All the children of level-1 nodes collectively form the \emph{level-2 nodes}. This process continues iteratively until we reach level $b$, at which point we obtain a depth-$b$ tree $T$. Finally, for each $k \leq b$, we refer to the rightmost level-$k$ node as $v_k$, and then the rightmost path of $T$ can be represented as $\bk{v_0, \ldots, v_b}$. See \cref{fig:obfuscating_tree}.

\begin{figure}[htbp]
    \centering
    \hspace{2cm}
    \scalebox{0.9}{\input{figures/figure_obfuscating_tree}}
    \caption{Obfuscating tree of depth $b = 3$. Each node has a randomly chosen number of children from the range $[1,3]$. The rightmost path, represented as $(v_0, v_1, v_2, v_3)$, has edge labels $X_1, X_2, X_3$.}
    \label{fig:obfuscating_tree}
\end{figure}
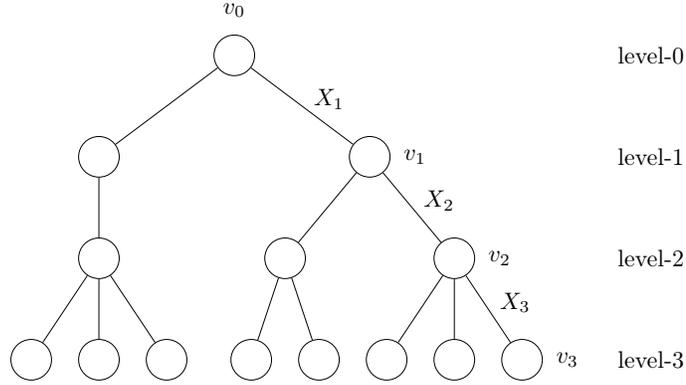

Next, we label each edge of $T$ with a sequence of keys. For each $k \leq b$, the rightmost edge $\BK{v_{k-1},v_k}$ is labeled with Alice's keys $X_k$. For any other edge $e$ that connects a level-$(k-1)$ node to a level-$k$ node---referred to as a \emph{level-$k$} edge---we independently sample a sequence $Y_k^{(e)} = (y_1, \dots, y_{n/b})$ of distinct keys from $U_k$ and assign $Y_k^{(e)}$ to $e$. As a result, each level-$k$ edge is labeled with a sequence of $n/b$ distinct keys from $U_k$.

The edge-labeled tree $T$ naturally induces an operational sequence $\sigma$, referred to as the \defn{obfuscation sequence}, as follows. The sequence $\sigma$ is defined according to a depth-first search (DFS) traversal of $T$. Whenever the traversal moves down along an edge, we insert the keys that appear in the edge label, and whenever it moves up along an edge, we delete the keys that appear in the edge label. 
In particular, when the traversal moves to a level-$k$ node $v$ of the tree, the corresponding prefix of $\sigma$ will result in a filter with its true key set being the union of edge labels along the path from $v$ to the root $v_0$ (i.e., consisting of $n/b$ keys from each of $U_1, U_2, \ldots, U_k$).

Throughout the protocol $Q_s$, the filters $F$ and $G$ will always have their operational histories as prefixes of $\sigma$. Specifically, for each $k \le b$, define $\sigma_{G_k}$ as the prefix of $\sigma$ that terminates just before inserting $X_{k+1}$ (corresponding to the traversal reaching the point where all children of $v_k$, except for $v_{k+1}$, have been visited, and the traversal is about to move to $v_{k+1}$). As a special case, $\sigma_{G_b}$ terminates just after inserting $X_b$. Similarly, for each $k \le b$, define $\sigma_{F_k}$ as the prefix of $\sigma$ that terminates just before deleting $X_k$ (corresponding to the traversal just before moving up from $v_k$ to $v_{k-1}$). See \cref{fig:obfuscating_sequence}.

Let $F_k$ and $G_k$ denote the filter states obtained after applying the operational sequence $\sigma_{F_k}$ and $\sigma_{G_k}$ on the initial empty filter, respectively. Then, both $F_k$ and $G_k$ will have their true key sets being the union of $X_1, \ldots, X_k$. Moreover, for any $\l \le r$, it is easy to check that $\sigma_{G_\l}$ is a prefix of $\sigma_{F_r}$ and their difference is a self-contained operational sequence, which implies $\accept{G_{\l}} \subseteq \accept{F_r}$ by the monotonicity assumption. Similarly, it is easy to check that the accepted set $\accept{G_{k}}$ is non-decreasing in $k$.\footnote{Note that, unlike $\accept{G_k}$, we cannot conclude that $\accept{F_k}$ is non-decreasing in $k$. This is because $\sigma_{F_k}$ is not necessarily a prefix of $\sigma_{F_{k+1}}$---in fact, the latter is shorter than the former. Fortunately, the monotonicity of $G_{k}$ will suffice for our purposes.} 

\begin{figure}[htbp]
    \centering
    \begin{subfigure}{0.45\textwidth}
        \centering
        \scalebox{0.9}{\input{figures/figure_sigma_G}}
        \caption{Operational sequence $\sigma_{G_1}$.}
        \label{fig:sigma_G}
    \end{subfigure}
    \hfill
    \begin{subfigure}{0.45\textwidth}
        \centering
        \scalebox{0.9}{\input{figures/figure_sigma_F}}
        \caption{Operational sequence $\sigma_{F_1}$.}
        \label{fig:sigma_F}
    \end{subfigure}
    \caption{Operational sequences $\sigma_{G_1}$ and $\sigma_{F_1}$. The sequence $\sigma_{G_1}$ follows a DFS traversal of $T$ and stops before traversing $(v_1, v_2)$, while the sequence $\sigma_{F_1}$ terminates just after traversing $(v_2, v_1)$.}
    \label{fig:obfuscating_sequence}
\end{figure}
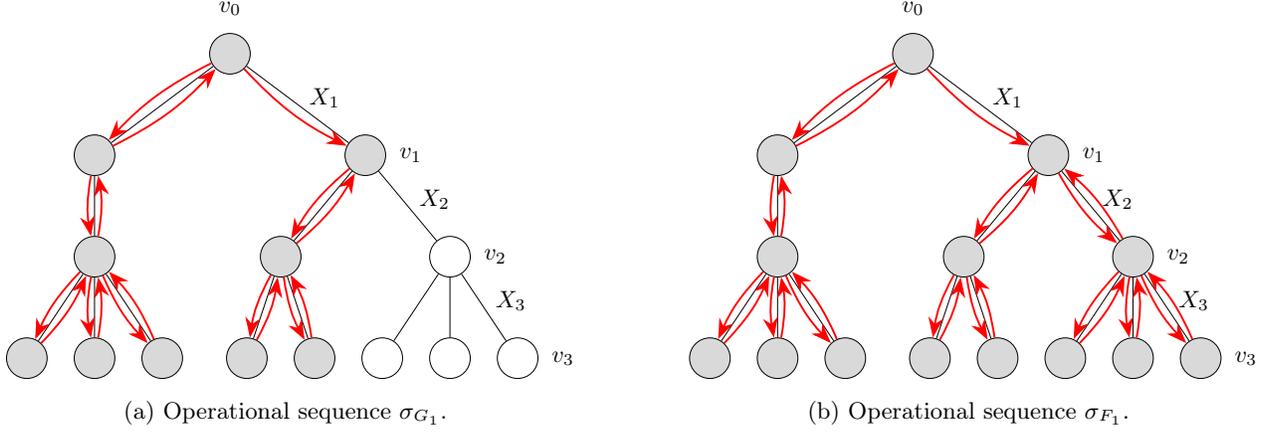

Based on these obfuscated filter states, the protocol $Q_s$ is defined as follows in \cref{alg:history_dependent_protocol}.

\begin{algorithm}[htbp]
  \caption{Protocol $Q_s$}
  \label{alg:history_dependent_protocol}
  \SetKwRepeat{Do}{do}{while}
    \SetKwFunction{Send}{Send} 
    \SetKwProg{Fn}{Subroutine}{:}{}
  
  \DontPrintSemicolon
    \Fn(\tcp*[f]{$X_k$ is contained in $F_r$'s true key set, but disjoint with $G_\l$'s.}){\Send{$X_k, F_r, G_\l$}}{
        \For{each $x$ in $X_k$} {
            $Z_x \longleftarrow$ indicator for whether $x \in \accept{G_\l}$\;
            Alice sends $Z_x$ to Bob\;
            \uIf {$Z_x = 0$ (i.e., $x \notin \accept{G_\l}$)} {
                Alice sends $x$ as an element in $\bk*{\accept{F_r} \setminus \accept{G_\l}} \cap U_k$\;
            } \Else{
                Alice sends $x$ as an element in $U_k$\;
            }  
        }
    }
    Alice prepares $F_b$ and sends it to Bob\;
    Bob prepares $G_0$ by himself\;
    \For{$k$ from $1$ to $s$} { 
        Alice sends $X_k$ using \Send{$X_k, F_{b}, G_{k-1}$}\;
        Bob computes $G_{k}$ using $X_k$ and $G_{k-1}$\;
    }
    \For{$k$ from $b$ to $s + 1$} {
        Alice sends $X_k$ using \Send{$X_k, F_k, G_s$}\;
        Bob computes $F_{k-1}$ using $X_k$ and $F_{k}$\;
    }
\end{algorithm}

Note that the obfuscating tree $T$, along with all its edge labels (except for the rightmost edges), is sampled using only public randomness. Hence, Bob can construct $G_0$ by executing $\sigma_{G_0}$ on an initially empty filter without any input from Alice. More generally, as long as Bob knows $X_k$, he can update the filter $G_{k-1}$ to $G_k$ by applying the difference between $\sigma_{G_k}$ and $\sigma_{G_{k-1}}$; likewise, as long as Bob knows $X_k$, he can also update the filter $F_k$ to $F_{k-1}$ by applying the difference between $\sigma_{F_{k-1}}$ and $\sigma_{F_k}$.

Below, we analyze the communication cost of the protocol $Q_s$ and determine a suitable choice of $s \in [0,b]$ to establish the desired lower bound for the filter.

\paragraph{Cost of sending a single key.}
In the main part of the protocol $Q_s$, Alice transmits each key to Bob using $\send{X_k, F_r, G_{\l}}$, where $\l < k \le r$ (meaning that $X_k$ is part of $F_r$'s true key set but is disjoint from $G_\l$'s). We begin by upper bounding the message entropy incurred by the subroutine $\send{X_k, F_r, G_{\l}}$.

Similar to the proof of \cref{prop:warmup}, we define parameters $a_0, \ldots, a_b$ such that 
\begin{align*}
    a_k \defeq \frac{1}{{(1 - \eps) n + \eps |U|}} \cdot {\E\Bk*{\abs{\accept{G_k}} - \abs{\accept{G_{k-1}}}}},
\end{align*}
where the expectation is taken over all sources of randomness, including the random tapes and the randomly sampled keys $X_1, \ldots, X_b$.
Since $\accept{G_k}$ is increasing in $k$, we get $a_k \ge 0$ for each $k$. Moreover, we have
\begin{align*}
    \E\Bk*{\abs{\accept{G_k}}} = \bk*{(1 - \eps) n + \eps |U|} \cdot a_{[0,k]}
\end{align*}
for each $k$, and $a_{[0,b]} = \E\Bk*{\abs{\accept{G_b}}} / \bk*{(1 - \eps) n + \eps |U|} \le 1$. Below, we establish an upper bound for message entropy of $\send{X_k, F_r, G_{\l}}$ in terms of $a_k$'s.

\begin{lemma}
    \label{lem:cost_for_single_X_k_dependent}
    Suppose parameters $b, M$ satisfy $M = 4^b$, $b = \omega(1)$, and $9^{b^2} = o( \eps |U| / n)$.\footnote{Recall that $|U| = \omega(n \eps^{-1})$. Hence, $\eps|U|/ n = \omega(1)$ and we can choose a superconstant $b$ such that $9^{b^2} = o( \eps |U| / n)$.}
    Then, the message entropy incurred by $\send{X_k, F_r, G_\l}$ is at most 
    \begin{align*}
        \frac{n}{b} \cdot \log \bk{|U|/b} +  \frac{n}{b} \log \eps + \bk{1 - \eps} \frac{n}{b} \cdot {\log { \bk*{a_{(\l, r]} + \frac{2}{4^b}}}} + o\bk*{\frac{n}{b}}.
    \end{align*}
\end{lemma}

\begin{proof}
    Since the protocol $\send{X_k, F_r, G_\l}$ sends each key $x$ in $X_k$ separately, we only need to upper bound the cost of sending each single key $x$.

    Similar to \cref{prop:warmup}, to send a key $x$, Alice must first send a biased bit $Z_x \defeq \ind\Bk{x \in \accept{G_\l}}$, which has entropy at most $h(\eps)$. Then, depending on the value of $Z_x$, she either sends $\log \abs*{\bk*{\accept{F_r}\setminus \accept{G_\l}} \cap U_k}$ bits (in the common case where $Z_x = 0$) or $\log \abs{U_k}$ bits (in the rare case where $Z_x = 1$). Thus, the total entropy is at most
    \begin{align*}
        &h(\eps) + \E\Bk*{Z_x \log |U_k| + \bk{1 - Z_x} \log \abs*{\bk*{\accept{F_r}\setminus \accept{G_\l}} \cap U_k}}\\
        {}\le{} & h(\eps) + \eps \log |U_k| + \bk{1 - \eps} \E\Bk*{\log \abs*{\bk*{\accept{F_r}\setminus \accept{G_\l}} \cap U_k} \mid x \notin \accept{G_\l}}\\
        {}\le{} & h(\eps) + \eps \log \bk{|U|/b} + \bk{1 - \eps} \log \E\Bk*{\abs*{\bk*{\accept{F_r}\setminus \accept{G_\l}} \cap U_k} \mid x \notin \accept{G_\l}} \\
        {}\le{} & h(\eps) + \eps \log \bk{|U|/b} + \bk{1 - \eps} \log \frac{\E\Bk*{\abs*{\bk*{\accept{F_r}\setminus \accept{G_\l}} \cap U_k}}}{ \Pr\Bk*{x \notin \accept{G_\l}}}, \numberthis\label{ineq:entropy_single_key_dependent}
    \end{align*}
    where the second inequality is due to the concavity of the logarithm, and the third inequality uses a similar technique as \cref{clm:entropy_single_key_independent}: For any non-negative integer random variable $X$ and event $\mathcal{E}$, we have 
    \begin{align*}
        \E\Bk{X \mid \mathcal{E}} 
        = \sum_{k=0}^\infty k \cdot \Pr\Bk{X = k \mid \mathcal{E}} 
        \le \sum_{k=0}^\infty k \cdot \frac{\Pr\Bk{X = k}}{\Pr\Bk{\mathcal{E}}}
        = \frac{\E\Bk{X}}{\Pr\Bk{\mathcal{E}}}. 
    \end{align*}
    
    Below, we further simplify $\E\Bk*{\abs*{\bk*{\accept{F_r}\setminus \accept{G_\l}} \cap U_k}}$ in terms of $a_k$'s.
    This is achieved by the following two claims. At a high level, the first claim says that, when evaluating $\E\Bk*{\abs*{\bk*{\accept{F_r}\setminus \accept{G_\l}} \cap U_k}}$, we can, to a first approximation, treat $U_k$ as independent of  $\accept{F_r}\setminus \accept{G_\l}$. The second claim then says that, thanks to the obfuscation performed within the obfuscation tree, the expectation of $|\accept{F_r}|$ is guaranteed to be almost the same as that of $|\accept{G_\l}|$. Once we have these claims in place, we will be able to complete the proof using essentially the same framework as in \Cref{prop:warmup}.

    \begin{claim}
        \label{clm:removing_Uk}
        For any $\l, r,k$ with $0\le \l < k \le r \le b$, we have
        \[
            \E\Bk*{\abs*{\bk*{\accept{F_r}\setminus \accept{G_\l}} \cap U_k}} 
            \le \frac{|U|}{b |U| - n M^{b+1}} \E\Bk*{\abs*{{\accept{F_r}\setminus \accept{G_\l}}}} + \frac{nM^{b+1}}{b}.
        \]
    \end{claim}
    \begin{proof}
        We analyze the expectation by conditioning on the randomness of the obfuscation sequence $\sigma$ and the random tape $\tapeD$ used by the algorithm $\calD$. 
        
        Fix any possible choices $\sigma^*$ and $\tapeD^*$. For the rest of the proof of the claim, we condition on $\sigma = \sigma^*$ and $\tapeD = \tapeD^*$. This means that the filter states $F_r$ and $G_\l$ are completely fixed, and so is the set $\accept{F_r} \setminus \accept{G_\l}$. Let $A^*$ denote the set $\accept{F_r} \setminus \accept{G_\l}$.

        On the other hand, the distribution of the set $U_k$ under this condition can be described as follows. First, conditioning on $\tapeD^*$ does not affect the distribution of $U_k$, because $U_k$ and $\sigma$ are determined by another public random tape, which is independent of $\tapeD$. Furthermore, conditioning on $\sigma^*$ slightly alters the distribution of $U_k$ in the following way. Each key that appears in $\sigma^*$ (i.e., in one of the edge labels of the obfuscating tree) has its membership in $U_k$ determined by $\sigma^*$, based on the depth of the corresponding edge in the obfuscating tree. Suppose there are $L$ such keys in total, of which $L_k$ belong to $U_k$. Then, outside these $L$ keys, $U_k$ behaves as a uniformly random subset of size $|U|/b - L_k$. In particular, for each key that is not present in $\sigma^*$, it belongs to $U_k$ with probability exactly $\frac{|U|/b - L_k}{|U| - L}$.

        Based on these observations, the expectation of $\abs*{\bk*{\accept{F_r}\setminus \accept{G_\l}} \cap U_k}$ can be upper bounded as follows:
        \begin{align*}
            &\E\Bk*{\abs*{\bk*{\accept{F_r}\setminus \accept{G_\l}} \cap U_k} \mid \sigma^*, \tapeD^*} 
            = \E\Bk*{\abs*{A^*\cap U_k} \mid \sigma^*, \tapeD^*} \\
            {}={}& \sum_{y \in A^*} \Pr\Bk*{y \in U_k \mid \sigma^*, \tapeD^*}
            {}\le{} L_k + \frac{|U|/b - L_k}{|U| - L} \abs*{A^*}. \numberthis \label{ineq:expectation_of_intersection_conditioned}
        \end{align*}
        Moreover, as the obfuscating tree has at most $M^b + M^{b-1} + \cdots + 1 < M^{b+1} / 2$ edges, we have $L_k \le L < n M^{b+1}/b$. Plugging this into \eqref{ineq:expectation_of_intersection_conditioned}, we further obtain
        \begin{align*}
            \E\Bk*{\abs*{\bk*{\accept{F_r}\setminus \accept{G_\l}} \cap U_k} \mid \sigma^*, \tapeD^*} 
            {}\le{} \frac{n M^{b+1}}{b} + \frac{|U|/b }{|U| - n M^{b+1}/b} \E\Bk*{\abs*{{\accept{F_r}\setminus \accept{G_\l}}} \mid \sigma^*, \tapeD^*}.
        \end{align*}
        Then, by taking expectation over $\sigma^*$ and $\tapeD^*$ on both sides, we obtain the desired inequality.
    \end{proof}

    \begin{claim}
        \label{clm:coupling}
        For any $k \le b$, we have
        \begin{align*}
            \E\Bk*{\abs*{\accept{F_k}}} \le \E\Bk*{\abs*{\accept{G_k}}} + \frac{(1 - \eps) n + \eps |U|}{M}.
        \end{align*}
    \end{claim}
    \begin{proof}
        Let $C_k$ be the number of children of the node $v_k$ in the obfuscating tree. The key point of the proof is to show that, for each integer $i$ with $1 \le i \le M-1$, we have 
        \begin{align*}
            \E\Bk*{\abs*{\accept{F_k}} \mid C_k = i} = \E\Bk*{\abs*{\accept{G_k}} \mid C_k = i + 1}. \numberthis \label{eq:coupling}
        \end{align*}
        \cref{eq:coupling} follows from the key observation that the distribution of $\bk*{\sigma_{F_k} \mid C_k = i}$ matches that of $\bk*{\sigma_{G_k} \mid C_k = i+1}$. 
        
        To establish this, recall that $\sigma_{F_k}$ is derived from the DFS traversal of the obfuscating tree $T$, which terminates after fully exploring the subtree rooted at $v_k$. Meanwhile, $\sigma_{G_k}$ is defined similarly, but its traversal halts after exploring all the subtrees rooted at $v_k$'s children except for the last one, which is rooted at $v_{k+1}$. Consequently, both $\bk*{\sigma_{F_k} \mid C_k = i}$ and $\bk*{\sigma_{G_k} \mid C_k = i+1}$ explore the first $i$ children of $v_k$, ensuring their corresponding traversals follow the same distribution.
        
        Moreover, note that the edge labels in $T$ are sampled independently, and every edge at the same level $j$ is drawn from the same distribution---whether it is the rightmost edge (labeled by $X_j$) or a non-rightmost edge (labeled by an independent sequence of $n/b$ distinct keys from $U_j$). It follows that $\bk*{\sigma_{F_k} \mid C_k = i}$ and $\bk*{\sigma_{G_k} \mid C_k = i+1}$ are identically distributed.

        Based on \eqref{eq:coupling}, we can conclude the desired inequality by 
        \begin{align*}
            \E\Bk*{\abs*{\accept{F_k}}}
            {}={}& \sum_{i = 1}^M \frac{1}{M} \E\Bk*{\abs*{\accept{F_k}} \mid C_k = i} \\
            {}={}& \sum_{i = 1}^{M-1} \frac{1}{M} \E\Bk*{\abs*{\accept{G_k}} \mid C_k = i + 1} + \frac{1}{M}\E\Bk*{\abs*{\accept{F_k}} \mid C_k = M}\\
            {}\le{}& \sum_{i = 1}^{M} \frac{1}{M} \E\Bk*{\abs*{\accept{G_k}} \mid C_k = i} + \frac{1}{M}\E\Bk*{\abs*{\accept{F_k}} \mid C_k = M}\\
            {}\le{}& \E\Bk*{\abs*{\accept{G_k}}} + \frac{(1 - \eps) n + \eps |U|}{M},
        \end{align*}
        where the last inequality uses that $\E_{\tapeD}\Bk*{\abs*{\accept{F_k}}} \le  (1 - \eps) n + \eps |U|$ since the filter algorithm $\calD$ only allows a positive error rate $\eps$.
    \end{proof}

    Combining \cref{clm:removing_Uk} and \cref{clm:coupling}, $\E\Bk*{\abs*{\bk*{\accept{F_r}\setminus \accept{G_\l}} \cap U_k}}$ can be further simplified as 
    \begin{align*}
        &\E\Bk*{\abs*{\bk*{\accept{F_r}\setminus \accept{G_\l}} \cap U_k}}  \\
        {}\le{}& \frac{|U|}{b |U| - n M^{b+1}} \E\Bk*{\abs*{{\accept{F_r}\setminus \accept{G_\l}}}} + \frac{nM^{b+1}}{b} \tag{by \cref{clm:removing_Uk}}\\
        {}={}& \frac{|U|}{b |U| - n M^{b+1}} \bk*{\E\Bk*{\abs*{\accept{F_r}}}- \E\Bk*{\abs*{\accept{G_\l}}}} + \frac{nM^{b+1}}{b}  \tag{by monotonicity}\\
        {}\le{}& \frac{|U|}{b |U| - n M^{b+1}} \bk*{\E\Bk*{\abs*{\accept{G_r}}}- \E\Bk*{\abs*{\accept{G_\l}}} + \frac{(1 - \eps) n + \eps |U|}{M}} + \frac{nM^{b+1}}{b} \tag{by \cref{clm:coupling}}\\
        {}={}& \frac{|U|}{b |U| - n M^{b+1}} \bk*{(1 - \eps) n + \eps |U|} \bk*{a_{(\l, r]} + \frac{1}{M}} + \frac{nM^{b+1}}{b}\\
        {}={}& \frac{\eps|U|}{b} \cdot \bk{1 + o(1)} \cdot \bk*{a_{(\l, r]} + \frac{1}{M}} + \frac{nM^{b+1}}{b},
        \numberthis \label{ineq:intersection_size_in_a_k}
    \end{align*}
    where the last line uses the following two approximations:
    \begin{itemize}
        \item Since $|U| = \omega(n\eps^{-1})$, we have $(1 - \eps)n + \eps |U| = \eps|U| \cdot \bk{1 + o(1)}$.
        \item Since $\eps|U|/n = \omega\bk*{9^{b^2}}$ and $M = 4^b$, we have $|U|/n = \omega\bk*{9^{b^2}} = \omega\bk*{M^{b+1}/b}$, which means $b |U| - n M^{b+1} = b|U| \cdot ( 1 - o(1))$.
    \end{itemize}
    Moreover, we have $\eps|U|/n = \omega\bk*{9^{b^2}} = \omega\bk*{4^b \cdot M^{b+1}}$, which means $n M^{b+1} < \eps |U| /4^b$. Hence, \eqref{ineq:intersection_size_in_a_k} can be further simplified as 
    \begin{align*}
         \frac{\eps|U|}{b} \cdot \bk{1 + o(1)} \cdot \bk*{a_{(\l, r]} + \frac{1}{M}} + \frac{nM^{b+1}}{b}
        {}\le {}& \frac{\eps|U|}{b} \cdot \bk{1 + o(1)} \cdot \bk*{a_{(\l, r]} + \frac{2}{4^b}}.
        \numberthis \label{ineq:intersection_size_in_a_k_simplified}
    \end{align*}
    Plugging \eqref{ineq:intersection_size_in_a_k_simplified} into \eqref{ineq:entropy_single_key_dependent}, the message entropy of sending a single key is further upper bounded by
    \begin{align*}
         & h(\eps) + \eps \log \bk{|U|/b} + \bk{1 - \eps} \log \frac{\E\Bk*{\abs*{\bk*{\accept{F_r}\setminus \accept{G_\l}} \cap U_k}}}{ \Pr\Bk*{x \notin \accept{G_\l}}} \\
        {}\le{} & o(1) + \eps \log \bk{|U|/b} + \bk{1 - \eps} \bk*{\log \bk*{\frac{\eps|U|}{b} \cdot \bk{1 + o(1)} \cdot \bk*{a_{(\l, r]} + \frac{2}{4^b}}} - \log \bk*{1 - \eps }} \\
        {}={} & \log \bk{|U|/b} + (1-\eps) \log \eps + \bk{1 - \eps} {\log { \bk*{a_{(\l, r]} + \frac{2}{4^b}}}} + o(1)\\
        {}={} & \log \bk{|U|/b} +  \log \eps + \bk{1 - \eps} {\log { \bk*{a_{(\l, r]} + \frac{2}{4^b}}}} + o(1).
    \end{align*}
    Hence, by summing up the communication cost for each key $x$ in $X_k$, we obtain the desired upper bound on the total message entropy incurred by $\send{X_k, F_r, G_\l}$:
    \begin{align*}
         &\frac{n}{b} \cdot \log \bk{|U|/b} +  \frac{n}{b} \log \eps + \bk{1 - \eps} \frac{n}{b} \cdot {\log { \bk*{a_{(\l, r]} + \frac{2}{4^b}}}} + o\bk*{\frac{n}{b}}. \qedhere
    \end{align*}
\end{proof}

\paragraph{Putting the pieces together.}
Now, to establish a lower bound for $\hfilter$, we analyze the total communication cost incurred by the protocol $Q_s$. Similar to \cref{prop:warmup}, the communication cost of $Q_s$ consists of the following two parts:
\begin{itemize}
    \item First, Alice transmits the filter state $F_b$, which requires $\hfilter$ bits.  
    \item Then, for each $k \leq b$, Alice sends $X_k$ using $\send{X_k, F_{r_k}, G_{\l_k}}$, where the parameters $\l_k$ and $r_k$ are defined as follows:  
    \begin{align*}
        \l_k \defeq  \begin{cases}
        k - 1 & \textup{for }k:  k \le s\\
        s & \textup{for }k: s+ 1 \le k \le b
    \end{cases} \quad \text{and} \quad
    r_k \defeq \begin{cases}
        b & \textup{for }k: k \le s \\
        k& \textup{for }k: s+1 \le k \le b
    \end{cases}.
    \end{align*}
\end{itemize}
We now choose parameters $M$ and $b$ such that $M = 4^b$, $b = \omega(1)$, and $9^{b^2} = o(\eps |U|/n)$. With this choice, the communication cost incurred by each $\send{X_k, F_{r_k}, G_{\l_k}}$ can be upper bounded using \cref{lem:cost_for_single_X_k_dependent}. Summing over all terms, the total message entropy of $Q_s$ is at most  
\begin{align*}
    \hfilter + n\log(|U|/b) + n\log \eps + (1 - \eps)\frac{n}{b} \cdot  \bk*{\sum_{k=1}^b \log \bk*{a_{(\l_k, r_k]} + \frac{2}{4^b}}} + o(n).
\end{align*}
Comparing to \cref{clm:message_entropy_lb_dependent}, we obtain
\begin{align*}
    \hfilter + n\log(|U|/b) + n\log \eps + (1 - \eps)\frac{n}{b} \cdot \bk*{\sum_{k=1}^b \log \bk*{a_{(\l_k, r_k]} + \frac{2}{4^b}}} + o(n)
    \ge b \log \bk*{(|U|/b)^{\underline{n/b}}}. \numberthis \label{ineq:filter_lb_dependent_without_simplifying}
\end{align*}
Moreover, as $|U| = \omega\bk{n}$, we have
\begin{align*}
    \bk*{\frac{|U|}{b}}^{\underline{n/b}}
    \ge \bk*{\frac{|U|}{b} - \frac{n}{b}}^{{n/b}}
    = \bk*{\frac{|U|}{b}}^{{n/b}} \bk*{1 - \frac{n}{|U|}}^{n/b}
    \ge \bk*{\frac{|U|}{b}}^{{n/b}} \bk*{1 - o\bk*{\frac{n}{b}}}.
\end{align*}
Therefore, the right-hand side of \eqref{ineq:filter_lb_dependent_without_simplifying} can be further simplified as
\begin{align*}
    b \log \bk*{\bk*{{|U|}/{b}}^{\underline{n/b}}}
    {}\ge{}& b \bk*{\log \bk*{\bk*{{|U|}/{b}}^{{n/b}}\bk*{1 - o\bk*{{n}/{b}}}}} \ge n \log ({|U|}/{b}) - o(n),
\end{align*}
and we obtain
\begin{align*}
    \hfilter \ge n\log \eps^{-1} - (1 - \eps)\frac{n}{b} \cdot \bk*{\sum_{k=1}^b \log \bk*{a_{(\l_k, r_k]} + \frac{2}{4^b}}} - o(n). \numberthis \label{ineq:filter_lb_dependent_simplified}
\end{align*}

Finally, we select an appropriate $s$, following a similar approach as in \cref{prop:warmup}, to derive the desired filter lower bound from \eqref{ineq:filter_lb_dependent_simplified}. By \cref{lem:choice_of_s}, there exists an $s$ with $0 \leq s \leq b$ such that
\begin{align*}
    \sum_{k=1}^b \log a_{(\l_k, r_k]} \le -b\log e + o(b). \numberthis \label{ineq:choice_of_s_dependent}
\end{align*}
Below, we show that such $s$ suffices to imply the desired filter lower bound.

\begin{lemma}
    \label{lem:removing_delta_in_log}
    For any real numbers $z_1, \ldots, z_b \in (0,1]$, suppose that 
    \[\sum_{k=1}^b \log z_k \le -b\log e + o(b).\numberthis \label{ineq:sum_of_log}\]
    Then, we have 
    \begin{align*}
        \sum_{k=1}^b \log \bk*{z_k + \frac{2}{4^b}} \le -b \log e + o(b). \numberthis \label{ineq:sum_of_log_with_delta}
    \end{align*}
\end{lemma}
\begin{proof}
    Depending on how large each $z_k$ is, there are two cases.

    \paragraph{Case 1.} Suppose there exists a $k^* \le b$, such that $z_{k^*} \le 3^{-b}$. Then, 
    \begin{align*}
        \log \bk*{z_{k^*} + \frac{2}{4^b}} < \log \frac{3}{3^b} < \log \frac{1}{e^b} = - b\log e.
    \end{align*}
    For any other $k \neq k^*$, we have 
    \begin{align*}
        \log \bk*{z_{k} + \frac{2}{4^b}}
        \le \log \bk*{1 + \frac{2}{4^b}} = o(1).
    \end{align*}
    Summing them together, we obtain \eqref{ineq:sum_of_log_with_delta}, as desired.

    \paragraph{Case 2.} Suppose for any $k \le b$, we have $z_k > 3^{-b}$. Then,
    \begin{align*}
        \log \bk*{z_{k} + \frac{2}{4^b}}
        = \log z_k + \log \bk*{1 + \frac{2}{4^b \cdot z_k}}
        < \log z_k + \log \bk*{1 + \frac{2}{(4/3)^b}}
        = \log z_k + o\bk*{1}.
    \end{align*}
    Summing over all $k$'s, we obtain 
    \begin{align*}
        \sum_{k=1}^b \log \bk*{z_k + \frac{2}{4^b}} 
        \le \sum_{k=1}^b \bk*{\log z_k + o(1)} 
        \le -b \log e + o(b),
    \end{align*}
    where the last inequality follows from \eqref{ineq:sum_of_log}.
\end{proof}

Combining \cref{lem:removing_delta_in_log} and \eqref{ineq:choice_of_s_dependent}, we get
\begin{align*}
    \sum_{k=1}^b \log \bk*{a_{(\l_k, r_k]} + \frac{2}{4^b}} \le -b\log e + o(b).
\end{align*}
Plugging this inequality into \eqref{ineq:filter_lb_dependent_simplified}, we obtain
\begin{align*}
    \hfilter 
    {}\ge{}& n\log \eps^{-1} - (1 - \eps) \cdot \frac{n}{b} \cdot \bk*{- b\log e + o(b)} - o(n) \\
    {}={}& n\log \eps^{-1} + n\log e - o(n).
\end{align*}
Finally, we verify that the lower-bound construction can be implemented using $f(n)$ operations, for any $f(n) = \omega(n)$. Since the obfuscation sequence $\sigma$ contains at most $n M^{b+1}/b$ operations, the protocol $Q_s$ requires at most $n M^{b+1}/b$ insertions and deletions to maintain the filters $F$ and $G$. By choosing $b = \omega(1)$ to grow sufficiently slowly, we can ensure that the number of operations is at most $f(n)$. This completes the proof of \cref{prop:history_dependent}.

%% file: figures/figure_obfuscating_tree.tex
\begin{tikzpicture}[
    level distance=1.5cm,
    level 1/.style={sibling distance=4cm},
    level 2/.style={sibling distance=2.5cm},
    level 3/.style={sibling distance=1cm},
    every node/.style={circle, draw, minimum size=0.6cm}
]
    \node[label=above:{$v_0$}] (root) {}
        child {node (l1) {}  %
            child {node (l2-1) {}
                child {node (l3-1) {}}
                child {node (l3-2) {}}
                child {node (l3-3) {}}
            }
        }
        child {node[label=right:{$v_1$}] (r1) {}
            child {node (r2-1) {}
                child {node (r3-1) {}}
                child {node (r3-2) {}}
            }
            child {node[label=right:{$v_2$}] (r2-2) {}
                child {node (r3-3) {}}
                child {node (r3-4) {}}
                child {node[label=right:{$v_3$}] (r3-5) {}}
            }
        };

    \node[draw=none, anchor=west] at (5.5, 0) {level-0};
    \node[draw=none, anchor=west] at (5.5, -1.5) {level-1};
    \node[draw=none, anchor=west] at (5.5, -3.0) {level-2};
    \node[draw=none, anchor=west] at (5.5, -4.5) {level-3};

    \node[draw=none] at ($(root)!0.5!(r1) + (0.4,0.1)$) {$X_1$};
    \node[draw=none] at ($(r1)!0.5!(r2-2) + (0.4,0.1)$) {$X_2$};
    \node[draw=none] at ($(r2-2)!0.5!(r3-5) + (0.4,0.1)$) {$X_3$};

\end{tikzpicture}

%% file: figures/figure_sigma_G.tex
\begin{tikzpicture}[
    level distance=1.5cm,
    level 1/.style={sibling distance=4cm},
    level 2/.style={sibling distance=2.5cm},
    level 3/.style={sibling distance=1cm},
    every node/.style={circle, draw, minimum size=0.6cm}
]
    \node[label=above:{$v_0$}, fill=gray!30] (root) {}
        child {node[fill=gray!30] (l1) {}  %
            child {node[fill=gray!30] (l2-1) {}
                child {node[fill=gray!30] (l3-1) {}}
                child {node[fill=gray!30] (l3-2) {}}
                child {node[fill=gray!30] (l3-3) {}}
            }
        }
        child {node[label=right:{$v_1$}, fill=gray!30] (r1) {}
            child {node[fill=gray!30] (r2-1) {}
                child {node[fill=gray!30] (r3-1) {}}
                child {node[fill=gray!30] (r3-2) {}}
            }
            child {node[label=right:{$v_2$}] (r2-2) {}
                child {node (r3-3) {}}
                child {node (r3-4) {}}
                child {node[label=right:{$v_3$}] (r3-5) {}}
            }
        };

    \node[draw=none] at ($(root)!0.5!(r1) + (0.4,0.1)$) {$X_1$};
    \node[draw=none] at ($(r1)!0.5!(r2-2) + (0.4,0.1)$) {$X_2$};
    \node[draw=none] at ($(r2-2)!0.5!(r3-5) + (0.4,0.1)$) {$X_3$};

    \draw[-{Stealth[length=3mm]}, red, line width=0.8pt, bend right=10] (root) to (l1);
    \draw[-{Stealth[length=3mm]}, red, line width=0.8pt, bend right=10] (l1) to (l2-1);
    \draw[-{Stealth[length=3mm]}, red, line width=0.8pt, bend right=10] (l2-1) to (l3-1);
    \draw[-{Stealth[length=3mm]}, red, line width=0.8pt, bend right=10] (l3-1) to (l2-1);
    \draw[-{Stealth[length=3mm]}, red, line width=0.8pt, bend right=10] (l2-1) to (l3-2);
    \draw[-{Stealth[length=3mm]}, red, line width=0.8pt, bend right=10] (l3-2) to (l2-1);
    \draw[-{Stealth[length=3mm]}, red, line width=0.8pt, bend right=10] (l2-1) to (l3-3);
    \draw[-{Stealth[length=3mm]}, red, line width=0.8pt, bend right=10] (l3-3) to (l2-1);
    \draw[-{Stealth[length=3mm]}, red, line width=0.8pt, bend right=10] (l2-1) to (l1);
    \draw[-{Stealth[length=3mm]}, red, line width=0.8pt, bend right=10] (l1) to (root);
    
    \draw[-{Stealth[length=3mm]}, red, line width=0.8pt, bend right=10] (root) to (r1);
    \draw[-{Stealth[length=3mm]}, red, line width=0.8pt, bend right=10] (r1) to (r2-1);
    \draw[-{Stealth[length=3mm]}, red, line width=0.8pt, bend right=10] (r2-1) to (r3-1);
    \draw[-{Stealth[length=3mm]}, red, line width=0.8pt, bend right=10] (r3-1) to (r2-1);
    \draw[-{Stealth[length=3mm]}, red, line width=0.8pt, bend right=10] (r2-1) to (r3-2);
    \draw[-{Stealth[length=3mm]}, red, line width=0.8pt, bend right=10] (r3-2) to (r2-1);
    \draw[-{Stealth[length=3mm]}, red, line width=0.8pt, bend right=10] (r2-1) to (r1);
        
\end{tikzpicture}

%% file: figures/figure_sigma_F.tex
\begin{tikzpicture}[
    level distance=1.5cm,
    level 1/.style={sibling distance=4cm},
    level 2/.style={sibling distance=2.5cm},
    level 3/.style={sibling distance=1cm},
    every node/.style={circle, draw, minimum size=0.6cm}
]
    \node[label=above:{$v_0$}, fill=gray!30] (root) {}
        child {node[fill=gray!30] (l1) {}  %
            child {node[fill=gray!30] (l2-1) {}
                child {node[fill=gray!30] (l3-1) {}}
                child {node[fill=gray!30] (l3-2) {}}
                child {node[fill=gray!30] (l3-3) {}}
            }
        }
        child {node[label=right:{$v_1$}, fill=gray!30] (r1) {}
            child {node[fill=gray!30] (r2-1) {}
                child {node[fill=gray!30] (r3-1) {}}
                child {node[fill=gray!30] (r3-2) {}}
            }
            child {node[label=right:{$v_2$}, fill=gray!30] (r2-2) {}
                child {node[fill=gray!30] (r3-3) {}}
                child {node[fill=gray!30] (r3-4) {}}
                child {node[label=right:{$v_3$}, fill=gray!30] (r3-5) {}}
            }
        };

    \node[draw=none] at ($(root)!0.5!(r1) + (0.4,0.1)$) {$X_1$};
    \node[draw=none] at ($(r1)!0.5!(r2-2) + (0.4,0.1)$) {$X_2$};
    \node[draw=none] at ($(r2-2)!0.5!(r3-5) + (0.4,0.1)$) {$X_3$};

    \draw[-{Stealth[length=3mm]}, red, line width=0.8pt, bend right=10] (root) to (l1);
    \draw[-{Stealth[length=3mm]}, red, line width=0.8pt, bend right=10] (l1) to (l2-1);
    \draw[-{Stealth[length=3mm]}, red, line width=0.8pt, bend right=10] (l2-1) to (l3-1);
    \draw[-{Stealth[length=3mm]}, red, line width=0.8pt, bend right=10] (l3-1) to (l2-1);
    \draw[-{Stealth[length=3mm]}, red, line width=0.8pt, bend right=10] (l2-1) to (l3-2);
    \draw[-{Stealth[length=3mm]}, red, line width=0.8pt, bend right=10] (l3-2) to (l2-1);
    \draw[-{Stealth[length=3mm]}, red, line width=0.8pt, bend right=10] (l2-1) to (l3-3);
    \draw[-{Stealth[length=3mm]}, red, line width=0.8pt, bend right=10] (l3-3) to (l2-1);
    \draw[-{Stealth[length=3mm]}, red, line width=0.8pt, bend right=10] (l2-1) to (l1);
    \draw[-{Stealth[length=3mm]}, red, line width=0.8pt, bend right=10] (l1) to (root);
    
    \draw[-{Stealth[length=3mm]}, red, line width=0.8pt, bend right=10] (root) to (r1);
    \draw[-{Stealth[length=3mm]}, red, line width=0.8pt, bend right=10] (r1) to (r2-1);
    \draw[-{Stealth[length=3mm]}, red, line width=0.8pt, bend right=10] (r2-1) to (r3-1);
    \draw[-{Stealth[length=3mm]}, red, line width=0.8pt, bend right=10] (r3-1) to (r2-1);
    \draw[-{Stealth[length=3mm]}, red, line width=0.8pt, bend right=10] (r2-1) to (r3-2);
    \draw[-{Stealth[length=3mm]}, red, line width=0.8pt, bend right=10] (r3-2) to (r2-1);
    \draw[-{Stealth[length=3mm]}, red, line width=0.8pt, bend right=10] (r2-1) to (r1);
    
    \draw[-{Stealth[length=3mm]}, red, line width=0.8pt, bend right=10] (r1) to (r2-2);
    \draw[-{Stealth[length=3mm]}, red, line width=0.8pt, bend right=10] (r2-2) to (r3-3);
    \draw[-{Stealth[length=3mm]}, red, line width=0.8pt, bend right=10] (r3-3) to (r2-2);
    \draw[-{Stealth[length=3mm]}, red, line width=0.8pt, bend right=10] (r2-2) to (r3-4);
    \draw[-{Stealth[length=3mm]}, red, line width=0.8pt, bend right=10] (r3-4) to (r2-2);
    \draw[-{Stealth[length=3mm]}, red, line width=0.8pt, bend right=10] (r2-2) to (r3-5);
    \draw[-{Stealth[length=3mm]}, red, line width=0.8pt, bend right=10] (r3-5) to (r2-2);
    \draw[-{Stealth[length=3mm]}, red, line width=0.8pt, bend right=10] (r2-2) to (r1);
    
\end{tikzpicture}

%% file: non-monotone.tex
In the previous sections, we proved \cref{thm:filter_lb} under the monotonicity assumption. In this section, we extend the proof to work for general filters without assuming the accepted set is monotone. 

The high-level idea of the proof is to define a modified version of the accepted set (called the \defn{reconstructible set}) for a general filter---one that remains monotone while preserving all the necessary properties of an accepted set. Specifically, let $\acceptvar{F}$ denote the reconstructible set of a filter $F$. Then, $\acceptvar{F}$ needs to satisfy the following properties:
\begin{itemize}
    \item Every key in the true key set of $F$ must also be in the reconstructible set $\acceptvar{F}$.  
    \item Any key not in the true key set of $F$ is included in $\acceptvar{F}$ with probability at most $\eps$, taken over the randomness of $\calD$.
    \item  The reconstructible set is monotone in the sense that, for any $\l \le r$, it satisfies $\acceptvar{G_\l} \subseteq \acceptvar{G_r}$ and $\acceptvar{G_\l} \subseteq \acceptvar{F_r}$, where $F_k$'s and $G_k$'s are the filters defined in \cref{sec:history_dependent}.  
\end{itemize}
One can verify that the proof of \cref{prop:history_dependent} remains valid when all instances of the accepted set are replaced with the reconstructible set, as long as the listed properties hold. Thus, it suffices to construct a reconstructible set that satisfies these properties. 

The definition of the reconstructible set depends on the partition $\pi$ of the universe. Fix a partition $\pi$. Throughout this section, we consider only operational sequences that \emph{conform} to $\pi$, defined as follows. 
\begin{definition}
    We say an operational sequence $\sigma$ \defn{conforms} to a partition $\pi$ if it is self-contained and satisfies the following property. If we execute $\sigma$ on an initially empty filter, then: 
    \begin{itemize}
        \item Whenever the size of the true key set is in the range $[(k-1)n/b, kn/b)$, $\sigma$ inserts only keys from $U_k$.
        \item Whenever the size of the true key set is in the range $((k-1)n/b, kn/b]$, $\sigma$ deletes only keys from $U_k$. 
    \end{itemize}
\end{definition}

Clearly, the obfuscation sequence defined in \cref{sec:history_dependent} conforms to the partition $\pi$. Below, we define the reconstructible set for any filter state $F$ that arises from executing an operational sequence conforming to $\pi$ on the initial empty filter.

\begin{definition}[Reconstructible set]
    The \defn{reconstructible set} $\acceptvar{F}$ of a filter $F$ is defined as the set of all keys $x$ satisfying the following property: There exists a filter state $F'$ that arises from executing an operational sequence $\sigma'$ conforming to $\pi$ on the initial empty filter, such that
    \begin{itemize}
        \item $F'$ has the same memory representation as $F$.
        \item $F'$ has the same true key set size as $F$.
        \item $x$ belongs to the true key set of $F'$.
    \end{itemize}
    
\end{definition}

For a given key $x$, any operational sequence $\sigma'$ with the above properties is said to be a \defn{reconstruction sequence}. The keys in the reconstructible set are precisely those for which at least one reconstruction sequence exists.

To conclude the proof of \cref{thm:filter_lb}, we check this reconstructible set has the desired properties in the following lemma.

\begin{lemma}
    For any filter states $F$ and $G$ that arise from executing operational sequences $\sigma_F$ and $\sigma_G$ conforming to $\pi$ on the initial empty filter, respectively, the following holds:
    \begin{itemize}
        \item The true key set of $F$ is contained in $\acceptvar{F}$, and $\acceptvar{F}$ is contained in $\accept{F}$.
        \item Suppose $\sigma_G$ is a prefix of $\sigma_F$ and their difference is self-contained. If the true key set size of $G$ is a multiple of $n/b$, say $\l n/b$, then we have $\acceptvar{G} \subseteq \acceptvar{F}$.
    \end{itemize} 
\end{lemma}
\begin{proof}
    First, it is easy to verify that the true key set of $F$ is contained in $\acceptvar{F}$. This is because the sequence $\sigma_F$ serves as the reconstruction sequence for each key in the true key set of $F$, ensuring that all such keys are included in $\acceptvar{F}$.

    Second, for any key $x \in \acceptvar{F}$, we show that $x$ is also in $\accept{F}$. By the definition of the reconstructible set, there exists a sequence $\sigma'$ conforming to $\pi$, such that the filter $F'$ resulting from the operational history $\sigma'$ has the same memory representation as $F$, and $x$ belongs to the true key set of $F'$. Since the true key set of $F'$ is contained in the accepted set $\accept{F'}$, we have $x \in \accept{F'}$. Moreover, since $F'$ shares the same memory representation with $F$, the query algorithm over $F'$ behaves exactly the same as over $F$, implying $\accept{F'} = \accept{F}$. Therefore, combining these two results, we conclude $x \in \accept{F}$.

    Finally, we show the reconstructible set is monotone. For any $x \in \acceptvar{G}$, we show that $x$ is also in $\acceptvar{F}$. Let $\sigma'_G$ be the reconstruction sequence of $x$ in $G$. Then, we claim that the operational sequence $\sigma_F'$, defined by concatenating $\sigma_G'$ with the difference $\sigma_F - \sigma_G$, is a valid reconstruction sequence of $x$ in $F$.

    We first verify that $\sigma_F'$ is self-contained.  Let $G'$ be the filter state resulting from the operational history $\sigma_G'$. Since both $\sigma_G'$ and $\sigma_F - \sigma_G$ are self-contained, the only way $\sigma_F'$ could fail to be self-contained is if there exists a key $x$ that is in the true set for $G'$, but is then inserted again in $\sigma_F - \sigma_G$ without first being deleted. Let us assume that this happens, and derive a contradiction.

    By the properties of the reconstruction sequence, the true key set size of $G'$ matches that of $G$, which is $\l n/b$. Since $\sigma_G'$ conforms to $\pi$, all keys in the true key set of $G'$ must belong to $\bigcup_{k=1}^\l U_k$, meaning $x \in \bigcup_{k=1}^\l U_k$. 

    On the other hand, because $\sigma_F - \sigma_G$ is self-contained, all insertions in $\sigma_F - \sigma_G$ occur when the true key set size is at least $\l n/b$. By the definition of conformity to $\pi$, these insertions must be from $\bigcup_{k=\l+1}^b U_k$. This contradicts the assumption that $x$ is inserted again in $\sigma_F - \sigma_G$, and completes the proof that $\sigma_F'$ is self-contained.

    Now, since both $\sigma_G'$ and $\sigma_F - \sigma_G$ conform to $\pi$, it is easy to further check that $\sigma_F'$ also conforms to $\pi$. Therefore, $\sigma_F'$ is a valid reconstruction sequence of $x$ in $F$, meaning that $x \in \acceptvar{F}$, as desired.
\end{proof}

\section{Open Questions}

We conclude with two directions for future work.

\paragraph{Tight upper and lower bounds for $\epsilon^{-1} = \Theta(1)$.} The lower bounds in this paper establish that, for $\epsilon = o(1)$, fingerprint filters achieve redundancy within $o(n)$ bits of optimal for any dynamic filter. 
    
When $\epsilon^{-1} = \Theta(1)$, we conjecture that fingerprint filters should continue to be optimal, but with the following caveat. When $\epsilon^{-1} = \Theta(1)$, the fingerprint filter can no longer ignore the fact that the fingerprints it is storing form a \emph{multiset} rather than a set. Any space-optimal implementation would have to encode this multiset in an information-theoretically optimal way, based on the distribution that the multiset comes from. Even from an upper-bound perspective, constructing such a fingerprint filter---while supporting time-efficient operations---remains open. 

Equally interesting is to prove strong lower bounds for the case of $\epsilon^{-1} = \Theta(1)$. This appears to be considerably more difficult than the case of $\epsilon^{-1} = \omega(1)$. It does not appear, for example, that the current proofs can, with more careful bookkeeping, be extended to get tight bounds in this regime. 

\paragraph{Tight lower bounds for value-dynamic retrieval. }Closely related to the problem of dynamic filters is the problem of \emph{value-dynamic retrieval} \cite{kuszmaul2024space}: Here, the task is to encode a function $f:S \rightarrow [2^v]$ for a set $S \subseteq U$ of $n$ keys, while supporting queries of the form ``return $f(s)$'' and updates of the form ``update $f(s) := x$'', where in both cases the user is responsible for guaranteeing that $s \in S$. 

The classic way to solve value-dynamic retrieval is with a minimal perfect hash function \cite{fredman1984size,hagerup2001efficient}, which necessarily incurs $n \log e - o(n)$ bits of redundancy. Kuszmaul and Walzer \cite{kuszmaul2024space} prove that $\Omega(n)$ bits of redundancy are necessary. They also show that, when $v = O(1)$, it is possible to construct an upper bound that achieves redundancy $n \log e - \Omega(n)$ bits. However, for the common case of $v = \omega(1)$, they conjecture that minimal perfect hashing should be optimal---that is, that any solution must incur $n \log e - o(n)$ bits of redundancy. 

In \cite{kuszmaul2024space}, the high-level approach that they take in their lower bounds is essentially the same for dynamic filters and for value-dynamic retrieval. We remark, however, that this connection appears to disintegrate as one tries to prove sharp bounds for the redundancy. The techniques in the current paper do \emph{not} appear to meaningfully extend to the value-dynamic retrieval problem. This is in part due to the following challenges:
\begin{itemize}
    \item \textbf{No ``empty states'':} In our lower bound for filters, we rely implicitly on the ability to send \emph{two} filters from Alice to Bob, where one of them is the empty filter, which can therefore be sent for free. For value-dynamic retrieval, there is no analogue to an empty filter---if Alice starts at data structure $A$ and performs updates to get to data structure $B$, then Bob cannot recover either $A$ or $B$ for free (unless Bob happens to already know $S$).
    \item \textbf{Avoiding inoperable states:} In value-dynamic retrieval, if the user ever performs an illegal update (updates $f(u)$ for some $u \notin S$), then the resulting data structure need not offer any guarantees of any sort (for either queries or updates). This significantly limits the types of operation sequences that Alice can use in constructing her communication protocol, and seems to prevent Alice from being able to rely on key-set changes in order to communicate information.
\end{itemize}
Based on these difficulties, we suspect that an entirely different approach will be needed if one is to prove Kuszmaul and Walzer's conjecture \cite{kuszmaul2024space}. Such a lower bound would be of great interest.